\documentclass[11pt]{article}
\usepackage{amsmath}
\usepackage{amssymb}
\usepackage{mathtools}
\usepackage{array}
\usepackage{amsthm}
\usepackage{authblk}
\usepackage{threeparttable}
\usepackage{rotating}
\usepackage{algorithm}
\usepackage{algpseudocode}
\usepackage{multirow}
\usepackage[a4paper, total={160mm,247mm},
 left=25mm,
 top=25mm]{geometry}

\newtheorem{theorem}{Theorem}
\newtheorem{proposition}[theorem]{Proposition}%
\newtheorem{lemma}[theorem]{Lemma}
\newtheorem{corollary}{Corollary}
\newtheorem{conjecture}{Conjecture}

\theoremstyle{remark}%
\newtheorem{example}{Example}%
\newtheorem{remark}{Remark}%

\theoremstyle{definition}%
\newtheorem{definition}{Definition}%
\newtheorem{construction}{Construction}

\title{In search of maximum non-overlapping codes}
\date{}

\author[1]{Lidija Stanovnik}
\author[1]{Miha Moškon}
\author[1]{Miha Mraz}

\affil[1]{University of Ljubljana, Faculty of Computer and Information Science, Večna pot 113, Ljubljana, Slovenia}

\begin{document}
\maketitle

\abstract{Non-overlapping codes are block codes that have arisen in diverse contexts of computer science and biology.
Applications typically require finding non-overlapping codes with large cardinalities,
    but the maximum size of non-overlapping codes has been determined only for cases where the codeword length divides the size of the alphabet,
    and for codes with codewords of length two or three.
For all other alphabet sizes and codeword lengths no computationally feasible way to identify non-overlapping codes that attain the maximum size has been found to date.
Herein we characterize maximal non-overlapping codes. 
We formulate the maximum non-overlapping code problem as an integer optimization problem and determine necessary conditions for optimality of a non-overlapping code.
Moreover, we solve several instances of the optimization problem to show that the hitherto known constructions do not generate the optimal codes for many alphabet sizes and codeword lengths.
We also evaluate the number of distinct maximum non-overlapping codes.
}


\section{Introduction}\label{intro}
\emph{Non-overlapping codes}, also known under the terms \emph{strongly regular codes}~\cite{Levenshtein:1964},
\emph{cross-bifix-free codes}~\cite{Bajic:2004,Chee:2013}, \emph{mutually uncorrelated codes} (MU-codes)~\cite{Yazdi:2018,Levy:2018},
and \emph{strong comma-free codes}~\cite{Fimmel:2017}, are block codes with the property that no prefix of a codeword
occurs as a suffix of any not necessarily distinct codeword.
Over the past sixty years, they have provided solutions to various problems in different fields.
They construct a special class of finite automata~\cite{Levenshtein:1970}, provide codes for frame
synchronization~\cite{Bajic:2004}, and even build addresses for DNA-storage~\cite{Yazdi:2018,Levy:2018}.
For most purposes, non-overlapping codes with large cardinalities are favorable and therefore researchers
have developed several constructions of such codes~\cite{Bajic:2014,Blackburn:2015,Barcucci:2016,Chee:2013,Bilotta:2012}.

Identifying non-overlapping codes with the largest cardinality, however, appears to be a hard nut to crack and little progress has been done in this direction.
Chee et al.~\cite{Chee:2013} searched through binary codes of lengths up to sixteen to show that
known constructions of non-overlapping codes obtain the maximum size, but for larger alphabet sizes and code lengths,
their approach cannot be used due to its double exponential time complexity and exponential space complexity.
Later, Blackburn~\cite{Blackburn:2015} provided formulas for the size of maximum two- and three-letter non-overlapping codes,
and for the sizes of maximum $q$-ary $n$-letter non-overlapping codes where $n$ divides $q$.
The number of maximal two- and three-letter, and maximum three-letter non-overlapping codes was recently computed~\cite{Fimmel:2023}.

Fimmel et al.~\cite{Fimmel:2023} also proposed an algorithm that they believe generates all maximal non-overlapping codes.
After proving the statement holds, it could theoretically be used to find the maximum $q$-ary $n$-letter non-overlapping codes, but the approach is impractical due to its double exponential time and exponential space complexity.
Nevertheless, by studying some equivalence relations we show that the time complexity of their algorithm can be reduced to exponential and the
space complexity of their algorithm to polynomial in $n$.
We further prove some necessary conditions for optimality.
We employ the results to compute optimal values for the size
and number of maximum non-overlapping codes that were not yet known, and demonstrate that other existing
constructions fail to generate a maximum code for many alphabet sizes and codeword lengths.

This article is arranged as follows:
Section~\ref{definitions} introduces definitions and briefly presents existing methods to construct non-overlapping codes.
In Section~\ref{construction} we prove that the algorithm proposed by Fimmel et al. generates non-overlapping codes only, determine the sizes of the constructed codes, show that the algorithm produces all maximal non-overlapping codes and characterize when a code is maximal.
We formulate an integer optimization problem for non-overlapping codes of maximum size and determine necessary conditions for its optimality in Section~\ref{sec:optimization}.
Section~\ref{formula:4} provides exact formulas for the number of maximal four-letter non-overlapping codes, and the size and the number of maximum four-letter non-overlapping codes.
In Section~\ref{results} we compare the largest codes produced by existing constructions and the newly computed optimal solutions for small parameter values.

\section{Definitions and constructions}\label{definitions}
Throughout this paper, we let $\Sigma$ be a finite alphabet with $q \geq 2$ elements and $n$ a positive integer.

\begin{definition}
Words $c_1 = x_1\cdots x_n \in \Sigma^n$ and $c_2 = y_1\cdots y_n \in \Sigma^n$ are \emph{non-overlapping} if
\begin{align*}
    \forall k \in \{1,\dots,n-1\}:\quad &x_{n+1-k}\cdots x_n \neq y_1\cdots y_k \\
    \text{and } \quad &y_{n+1-k}\cdots y_n \neq x_1\cdots x_k.
\end{align*}
In other words, no prefix of $c_1$ occurs as a suffix of $c_2$ and vice versa.
\end{definition}
\begin{example}
The word VRV is not non-overlapping because V occurs as a prefix and a suffix.
The words VRT and KRT are non-overlapping because the prefixes V, VR, K and KR never occur as suffixes.
\end{example}

\begin{definition}
$X \subseteq \Sigma^n$ is a \emph{non-overlapping code} if all (not necessarily distinct) elements $c_1, c_2 \in X$ are non-overlapping.
\end{definition}

\begin{definition}
A non-overlapping code $X$ is \emph{maximal}, if it cannot be expanded, meaning that every $n$-letter word over $\Sigma$ that is not in $X$
either has a prefix that overlaps with a suffix in $X$ or a suffix that overlaps with a prefix in $X$.

\end{definition}
\begin{example}
The non-overlapping code $\{\text{VRT}, \text{KRT}\}$ is not maximal, because it could be expanded by adding the word RRT.
The non-overlapping code $\{\text{VRT}, \text{VVT}, \text{RVT}, \text{RRT}\}$ is maximal over the alphabet $\{\text{V}, \text{R}, \text{T}\}$,
because no other three-letter word exists that starts with V or R, ends with T, and includes neither prefix VT nor prefix RT.
\end{example}
\begin{definition}
A non-overlapping code $X$ is \emph{maximum} if for all non-overlapping codes
$Y \subseteq \Sigma^n$: \[\lvert X \rvert \geq \lvert Y \rvert.\]
We denote this greatest cardinality with $S(q,n)$ and define $N(q,n)$ to be the number of codes that achieve it.
\end{definition}
\begin{remark}
Every maximum non-overlapping code is maximal.
\end{remark}

Non-overlapping codes were first defined by Levenshtein~\cite{Levenshtein:1964} who found that they coincide with
a class of codes for which there exists a decoding automaton that correctly decodes a sequence of input letters
independently of the choice of the initial state of the automaton~\cite{Levenshtein:1970}.
This property ensures that errors in the input word and random transitions between states of the automaton do not affect the decoding of subsequent input words.
He proposed Construction~\ref{construction:levenshtein} (below) that generates non-overlapping codes of not necessarily optimal sizes.
He also established a lower bound $S(q,n) \gtrsim \frac{q-1}{qe} \frac{q^n}{n}$~\cite{Levenshtein:1964} and an upper bound
$S(q,n) \leq \left(\frac{n-1}{n}\right)^{n-1}\frac{q^n}{n}$~\cite{Levenshtein:1970} for the size of a maximum non-overlapping code.

\begin{construction}
    \label{construction:levenshtein}
    Let $n > 1$ and $q > 1$  be integers and $1 \leq k \leq n - 1$. Denote by $C$ the set of all codewords $s = (s_1,s_2,\dots,s_n) \in \mathbb{Z}_q^n$ such that $s_1 \neq 0$,  $s_{n-k} \neq 0$, $s_{n-k+1} = s_{n-k+2} = \cdots = s_n = 0$, and $(s_{1},\dots,s_{n-k})$ does not contain $k$ consecutive 0's. Then $C$ is a non-overlapping code.
\end{construction}

His construction was later rediscovered by Chee et al.~\cite{Chee:2013}, who reversed the codewords.
They showed that the size of a code $C$ obtained using this construction and parameter $ 2 \leq k \leq n - 2$ equals
$\lvert C \rvert = (q-1)^2 F_{k,q}(n-k-2)$ with initialization $F_{k,q} (i) = q^i\; \forall i \in \{0, \dots, k-1\}$
and a $(q-1)$-weighted $k$-generalized Fibonacci recurrence relation $F_{k,q}(n) = (q-1) \sum_{l=1}^k F_{k,q}(n-l)$.
Additionally, they searched for maximum binary non-overlapping codes with $n \leq 16$ by determining a maximum clique
in a graph with vertices that correspond to words in $\mathbb{Z}_2^n$ that are non-overlapping and edges that correspond
to pairs of non-overlapping words.
Note that this approach cannot be efficiently applied to larger alphabets and code lengths because computing a maximum
clique of a graph with $m$ vertices requires $O(2^{\frac{m}{4}})$ time~\cite{Robson:2001} and the graph has $O(q^n)$ vertices.

Another construction that generates binary non-overlapping codes was provided by Bilotta et al.~\cite{Bilotta:2012}.
It uses Dyck sequences as described in Construction~\ref{construction:bilotta}, and for both even length $2n+2$ and odd
length $2n+1$ it generates codes with $n$-th Catalan number of codewords, $\lvert C_{2n+2} \rvert = \lvert C_{2n+1} \rvert = \frac{1}{n+1}\binom{2n}{n}$.

\begin{construction}
    \label{construction:bilotta}
    Let $\mathcal{D}$ be a set of Dyck sequences of even length $2n$, i.e. binary sequences composed of $n$ zeros and
    $n$ ones such that no prefix of the sequences has more zeros than ones.
    The set of sequences of even length $C_{2n+2} = \{1a0: a \in \mathcal{D}\}$ and the set of sequences of odd length
    $C_{2n+1} = \{1a: a \in \mathcal{D}\}$ are non-overlapping codes.
\end{construction}

This class of codes is not restricted to binary codes as Wang and Wang~\cite{Wang:2021} showed that any binary
non-overlapping code can be generalized to $q$-ary non-overlapping codes using Construction~\ref{construction:wang} (below).
The size of the generalized codes was, however, not further theoretically analysed.
\begin{construction}
    \label{construction:wang}
    Let $S$ be a binary non-overlapping code and $(I, J)$ a partition of $\mathbb{Z}_q$ with $q \geq 2$.
    For $w = w_1\cdots w_n \in S$ we define a mapping $\Phi(w) = \{y_1\cdots y_n \mid y_i \in I\text{ if } w_i = 0\text{ and }y_i \in J\text{ if } w_i= 1\}$.
    The set $\Phi(S) = \bigcup_{w \in S} \Phi(w)$ is a non-overlapping code.
\end{construction}

Blackburn~\cite{Blackburn:2015} provided the class of non-overlapping codes described in Construction~\ref{construction:blackburn} (below)
and proved that it contains a maximum non-overlapping code when $n$ divides $q$, as a code with $k=n-1$ and $S=I^k$
reaches the upper bound set by Levenshtein.
He also proved that $S(q,2) = \lfloor \frac{q}{2} \rfloor \lceil \frac{q}{2} \rceil $ and
$S(q,3) = \left[\frac{2q}{3}\right]^2 \left(q - \left[\frac{2q}{3}\right]\right)$,
where $\left[x\right]$ denotes the rounding of $x$ to the closest integer.
He conjectured that for $n > 2$ and sufficiently large $q$, a maximum non-overlapping code is given by Construction~\ref{construction:blackburn}
with $k=n-1$ and some value of $l$, but this conjecture is yet to be (dis)proven.

\begin{construction}
    \label{construction:blackburn}
    Let $k$ and $l$ be such integers that $1 \leq k \leq n-1$ and $1 \leq l \leq q - 1$.
    Let $\Sigma = I \cup J$ be a partition of a set $\Sigma$ of cardinality $q$ into two parts $I$ and $J$ of cardinalities $l$ and $q-l$ respectively.
    Let $S \subseteq I^k \subseteq \Sigma^k$.
    Let $C_{S,J}(n)$ be the set of all words $c \in \Sigma^n$ such that $c_1 c_2 \cdots c_k \in S$, $c_{k+1} \in J$, $c_n \in J$
    and no subword of the word $c_{k+2}c_{k+3}\cdots c_{n-1}$ lies in $S$.
    Then $C_{S,J}(n)$ is a non-overlapping code.
\end{construction}
The size of a general code obtained using Blackburn's construction is still an open problem, but Wang and Wang~\cite{Wang:2021}
provided a formula for the case $S = I^k$.
It states that for $k \in \{n-1, n-2\}:\;$ $\lvert C_{I^k,J}(n)\rvert = \lvert I \rvert^k \lvert J \rvert^{n-k}$ and for
$k \leq n - 3:\;$ $\lvert C_{I^k,J}(n) \rvert = q \cdot \lvert C_{I^k,J}(n-1)\rvert - \lvert I \rvert^k \lvert J\rvert \cdot \lvert C_{I^k,J}(n-k-1)\rvert$.

Barcucci et al.~\cite{Barcucci:2016} defined a set of non-overlapping codes given in Construction~\ref{construction:barcucci} (below)
that are generated using colored Motzkin words and showed they are all maximal.
They did not compute the size of the codes.

\begin{definition}
    A sequence in $ \mathbb{Z}_{q+2}^n$ that contains the same number of zeros and ones, such that no prefix of the sequences
    has more zeros than ones is called a \emph{$q$-colored Motzkin word} of length $n$.
    We denote the set of all such sequences with $\mathcal{M}_q(n)$.
\end{definition}

\begin{definition}
An \emph{elevated $q$-colored Motzkin word} of length $n$ is a sequence $1\alpha0$, where $\alpha \in \mathcal{M}_{q}(n-2)$.
We denote the set of all such words with $\hat{\mathcal{M}}_{q}(n)$.
\end{definition}

\begin{construction}
    \label{construction:barcucci}
    $CBFS_q(n) = A_q(n) \cup B_q(n) \cup C_q(n)$ is a maximal non-overlapping code where
    $A_q(n) = \{\alpha\beta: \alpha\in \mathcal{M}_{q-2}(i), \beta \in \hat{\mathcal{M}}_{q-2}(n-i)\} \setminus \{\alpha\beta: \alpha, \beta \in \hat{\mathcal{M}}_{q-2}(\frac{n}{2})\}$,
    $B_q(n) = \{1\alpha\beta: \beta \in \mathcal{M}_{q-2}(i), \beta \in \hat{\mathcal{M}}_{q-2}(n-i-1)\}$, and
    $C_q(n) = \{\gamma 0: \gamma \in \mathcal{M}_{q-2}(n-1), \gamma \neq u\beta v, \beta \in \hat{\mathcal{M}}_{q-2}(j) \}$.
\end{construction}

Fimmel, Michel and Str{\"u}ngmann~\cite{Fimmel:2017} proposed that three-letter non-overlapping codes emerged naturally.
In living cells the ribosome behaves as a decoding automata that decodes a sequence of codons to a sequence of amino acids
using the genetic code as a decoding function.
Data suggests that ancestral genetic code was indeed a non-overlapping code.
Inspired by Blackburn's work, they recently showed that the set of maximal two-letter non-overlapping codes over $\Sigma$ is exactly the set of partitions of $\Sigma$ into two non-empty parts and counted that there are $2^q - 2$ such codes~\cite{Fimmel:2023}.
Moreover, they characterized maximal three-letter non-overlapping codes (see Proposition~\ref{proposition:fimmel} below) and counted that there are $\sum_{m=1}^{q-1} \binom{q}{m}2^{m(q-m)}$ such codes. They also determined $N(q,3) = 2\binom{q}{\left[\frac{2q}{3}\right]}$.

\begin{definition}
    Let $L$ and $R$ be sets of strings. \\
        (i)     $\left(LR\right)$ denotes the concatenation of sets $L$ and $R$, i.e. the set of all strings of the form $lr$, where $l \in L$ and $r \in R$, \\
        (ii) $L^i$ denotes the concatenation of the set $L$ with itself $i$ times, \\
        (iii) the Kleene star denotes the smallest superset that is closed under concatenation and includes the empty set $L^* = \bigcup_{i \geq 0} L^i$ and \\
        (iv) the Kleene plus denotes the smallest superset that is closed under concatenation and does not include the empty set  $L^+ = \bigcup_{i > 0} L^i$.
\end{definition}

\begin{proposition}[\cite{Fimmel:2023}, Theorem 5.3]\label{proposition:fimmel}
    The set of maximal three-letter non-overlapping codes is exactly the set of all three-letter codes $X_3 = (L_1R_2) \cup (L_2R_1)$, where \\ 
        (i) $(L_1, R_1)$ is a partition of $\Sigma$ into two non-empty parts, and \\
        (ii) $(L_2, R_2)$ is a partition of $(L_1 R_1)$.
\end{proposition}

\section{Constructing maximal non-overlapping codes}\label{construction}
Fimmel et al.~\cite{Fimmel:2023} show that a straightforward generalisation of Proposition~\ref{proposition:fimmel} does not hold.
The underlying construction can, however, be generalized as described in Construction~\ref{construction:fimmel} (below). They posit that it is composed of non-overlapping codes only and that it contains all maximal non-overlapping codes.

\begin{construction}
    \label{construction:fimmel}
    Let $n \geq 3$. We define $\mathcal{M}_{q,n}$ to be the set of all codes
    $C = \bigcup_{i=1}^{n-1} \left(L_i R_{n-i} \right) \subseteq \Sigma^n$, where \\
        (i) $(L_1, R_1)$ is a partition of $\Sigma$ into two non-empty parts, and \\
        (ii) $(L_i, R_i)$ is a partition of $\bigcup_{j=1}^{i-1} \left(L_j R_{i-j} \right)$ for every $i \in \{2,\dots, n-1\}$.
\end{construction}

Before proving both statements, let us observe that the mapping from the partitions of Construction~\ref{construction:fimmel} to the set of codes $C \subseteq \Sigma^n$ is not injective. Two distinct collections of partitions $(L_i,R_i)_{i=1,\dots,n-1}$ can determine the same code as demonstrated by Example~\ref{example:binary4}.
\begin{example} \label{example:binary4}
    Set $n=4$ and take a partition $(L_1,R_1)$ of $\Sigma$.
    If $L_2 = (L_1R_1)$, $R_2 =\emptyset$ and $L_3 = \emptyset$, $R_3 = (L_1R_2) \cup (L_2R_1) = (L_1R_1^2)$,
    we get a code $C_1 = (L_1R_3) \cup (L_2R_2) \cup (L_3R_1) = (L_1^2 R_1^2)$.
    Now observe  $L_2 = \emptyset$, $R_2 =(L_1R_1)$ and $L_3 = (L_1R_2) \cup (L_2R_1) = (L_1^2R_1)$, $R_3 = \emptyset$.
    We get a code $C_2 = (L_1R_3) \cup (L_2R_2) \cup (L_3R_1) = (L_1^2 R_1^2)$.
    In particular, $C_1 = C_2$.
\end{example}

\begin{proposition}
  \label{proposition:z}
    Let $C=\bigcup_{i=1}^{n-1} \left(L_i R_{n-i} \right) \in \mathcal{M}_{q,n}$.
    Define $X \coloneqq \bigcup_{i=1}^{2n+1} X_i$, such that
    \begin{align*}
        X_{2i-1} &\coloneqq L_i, \\
        X_{2i} &\coloneqq R_i, \\
        X_{2n+1} &\coloneqq C.
    \end{align*}
    No proper prefix in $X$ occurs as a proper suffix in $X$.
\end{proposition}

\begin{proof}
    Assume, for sake of contradiction, that the statement does not hold.
    Let $p$ be the length of the shortest proper prefix in $X$ that occurs as a proper suffix in $X$.
    Let $i$ be the smallest index such that there exists a word $w \in X_i$ with a prefix $x_1\cdots x_p$ that occurs as a suffix in $X$.
    Let $j$ be the smallest index such that $X_j$ contains a word $w'$ that ends in $x_1\cdots x_p$.
    Clearly $i,j > 2$, as $X_1$ and $X_2$ have no proper prefixes nor suffixes.
    Therefore $w \in (L_k R_{\lceil \frac{i}{2} \rceil - k}) \subseteq X_i$ for some $k < \lceil \frac{i}{2}\rceil$.
    
    If $k < p$, then $x_{k+1}\cdots x_p$ is a prefix in $R_{\lceil \frac{i}{2} \rceil - k} = X_{2{\lceil \frac{i}{2} \rceil} - 2k}$ and suffix in $X_j$ that is shorter than $p$.
    Since $\lceil \frac{i}{2} \rceil - k > 1$ (otherwise $x_1\cdots x_p = w$ is not a proper prefix of $w$), this contradicts the minimality of $p$.
    If $p < k$, then $x_1\cdots x_p$ is a prefix in $L_k = X_{2k-1}$.
    Since for $l > 2$ every one-letter prefix in $X$ is from $L_1$ and every one-letter suffix from $R_1$, $p > 1$ and $k > 2$.
    Therefore $2k - 1 > 1$ and $2k - 1 < i$. This contradicts minimality of $i$.
    So $p = k$ and hence $x_1\cdots x_p \in L_k = L_p$. Repeat the symmetric procedure on $X_j$ to obtain a word $w' \in (L_{\lceil \frac{j}{2} \rceil - l} R_l)$ ending in $x_1\cdots x_p$ and $p=l$.
    So $x_1\cdots x_p \in R_l = R_p$. Then $x_1\cdots x_p \in L_p \cap R_p$, but $L_p \cap R_p = \emptyset$ by definition.
\end{proof}

\begin{theorem}
  \label{theorem:c6}
    $C \in \mathcal{M}_{q,n}$ is a non-overlapping code.
\end{theorem}

\begin{proof}
    The theorem follows directly from Proposition~\ref{proposition:z}.
\end{proof}

\begin{corollary}
 \label{remark:l1r1}
    (i) $L_k \cup R_k$ is a non-overlapping code for all $1 \leq k \leq n - 1$.\\
    (ii) If $w_1\cdots w_k \in L_k \cup R_k$ for some $k > 1$, then $w_1 \in L_1$ and $w_k \in R_1$.
\end{corollary}
\begin{proof}
    The corollary follows directly from Construction~\ref{construction:fimmel} and Theorem~\ref{theorem:c6}.
\end{proof}

\subsection{Size of $C \in \mathcal{M}_{q,n}$}
Now that we know that every $C \in \mathcal{M}_{q,n}$ is non-overlapping, we want to determine the size of $C$. 
Theorem~\ref{prop:size} (below) shows that it depends on the sizes of the sets partitions $(L_i,R_i)_{i=1,\dots,n-1}$ only.
Before stating and proving the formula, we define a finite sequence of decompositions due to Fimmel et al.~\cite{Fimmel:2019} and prove some of its properties.

\begin{definition}
Let $w \in \bigcup_{i=1}^{n-1} \left(L_i R_{n-i} \right)$ and let us set $L \coloneqq \bigcup_{i<n} L_i$ and $R \coloneqq \bigcup_{i<n} R_i$.
    Define $\{p_k(w)\}_{k \geq 0}$ to be a sequence of decompositions of $w$ in $\left(L\left(L\cup R\right)^* R\right)$ encoded with a binary word ${p_k}$ over the alphabet $\{\text{l},\text{r}\}$ such that \\
        (i) $p_0(w) \in \left(\text{l}\{\text{l},\text{r}\}^{n-2}\text{r}\right)$ with $p_0(w)_i = \text{l}$ if $w_i \in L_1$ and $p_0(w)_i=\text{r}$ if $w_i \in R_1$,\\
        (ii) if $p_k(w) \in \left(\text{l}\{\text{l},\text{r}\}^{+}\text{r}\right)$, then $p_{k+1}(w) \in \left(\text{l}\{\text{l},\text{r}\}^* \text{r}\right)$ is obtained by replacing each occurrence of lr in $p_{k}(w)$ by l if lr corresponds to a subword $w'$ of $w$ in $(L_i R_j) \subseteq L_{i+j}$, or by r if it corresponds to a subword $w'$ of $w$ in $(L_i R_j) \subseteq R_{i+j}$ for some positive integers $i$ and $j$.
\end{definition}

\begin{proposition}\label{xy:LR}
Let $x \in L$ and $y \in R$.
If the length of $xy$ is at most $n-1$, then $xy \in L \cup R$.
\end{proposition}

\begin{proof}
    Since $x \in L$ there exists some $i$, $0 < i < n$, such that $x \in L_i$.
    Since $y \in R$ there exists some $j$, $0 < j < n$, such that $y \in R_j$.
    Therefore $xy \in (L_iR_j) \subseteq \bigcup_{k < i + j} (L_k R_{i+j-k})$.
    If $i + j < n$, the latter is partitioned into $(L_{i+j}, R_{i+j})$, so either $xy \in L_{i+j} \subseteq L$ or $xy \in R_{i+j} \subseteq R$.
\end{proof}

\begin{proposition}\label{proposition:pk}
    The sequence $\{p_k(w)\}_{k \geq 0}$ is \\
    (i) well-defined,\\
    (ii) finite and its last element is {\normalfont{lr}}.
\end{proposition}

\begin{proof}
    (i) Every letter in $w$ belongs to $\Sigma = L_1 \cup R_1$.
    If lr occurs as a proper substring of $p_k(w)$, then by Proposition~\ref{xy:LR} it corresponds to some $k$-letters long substring $w'$ of $w$, such that $w' \in L \cup R$.
    Therefore lr can be replaced by either l or r.
    If lr occurs at the beginning of $p_k(w)$, then it is replaced by l. Otherwise $w' \in R_i$ for some $i < n - 1$ and there is a word in $(L_1R_i) \in L \cup R$ that ends in $w'$ which contradicts Proposition~\ref{proposition:z}.
    A symmetric observation guarantees that a lr at the end of $p_k(w)$ is replaced by r.

    (ii) If $p_k(w) \in (\text{l}\{\text{l},\text{r}\}^+\text{r})$, then the length of $p_{k+1}(w)$ is strictly smaller then the length of $p_k(w)$.
    If $p_k(w) \not\in (\text{l}\{\text{l},\text{r}\}^+\text{r})$, then $p_{k+1}$ is not defined and $p_k(w)$ is the last element of the sequence.
    The previous step reveals that then $p_k(w) = \text{lr}$.
\end{proof}

\begin{corollary}\label{corrolary:pk}
    Let $1 < l < n$ and $w \in L_l \cup R_l$. \\
    The sequence $\{p_k(w)\}_{k \geq 0}$ is finite. The last element of the sequence is {\normalfont{lr}}.
\end{corollary}
\begin{proof}
    If $w$ belongs to $L_l$ (alternatively to $R_l$), then it also belongs to $L_l \cup R_l$. The latter is a non-overlapping code as noted in Corollary~\ref{remark:l1r1}, so the statement follows from Proposition~\ref{proposition:pk}.
\end{proof}

\begin{theorem}\label{prop:size}
    $\lvert \bigcup_{i=1}^{n-1} \left(L_i R_{n-i} \right) \rvert = \sum_{i=1}^{n-1} \lvert L_i\rvert \lvert R_{n-i}\rvert$.
\end{theorem}

\begin{proof}
    It is sufficient to explain that for every word $w \in \bigcup_{i=1}^{n-1} \left(L_i R_{n-i} \right)$ there is a unique pair of sets $L_i$ and $R_{n-i}$ such that
    $w \in \left(L_i R_{n-i} \right)$.

    Assume, for sake of contradiction, that there exists a word $w = w_1\cdots w_n \in \bigcup_{i=1}^{n-1} \left(L_i R_{n-i} \right)$ and indices $i < j$ such that $w \in (L_iR_{n-i})$ and $w \in (L_jR_{n-j})$.
    Observe the sequence $\{p_k(w_{i+1}\cdots w_j)\}_{k \geq 0}$.\\
    \textbf{STEP 1:} The sequence $\{p_k(w_{i+1}\cdots w_j)\}_{k \geq 0}$ is well-defined, i.e. $\forall k \geq 0 :\; p_k \in (\text{l}\{\text{l},\text{r}\}^*\text{r})$.

    \textbf{(i)} $p_0 (w_{i+1} \cdots w_j) \in (\text{l}\{\text{l},\text{r}\}^{j-i-2}\text{r})$.\\
        Since $1 \leq i < j$ and  $w_1 \cdots w_j \in L_j$, Corollary~\ref{remark:l1r1} implies $w_{j} \in R_1$. 
        $i+1 \leq j < n$ and $w_{i+1} \cdots w_n \in R_{n-i}$, so Corollary~\ref{remark:l1r1} implies $w_{i+1} \in L_1$.

    \textbf{(ii)} If $p_k(w_{i+1}\cdots w_j) \in (\text{l}\{\text{l},\text{r}\}^+\text{r})$, then $p_{k+1}(w_{i+1}\cdots w_j) \in (\text{l}\{\text{l},\text{r}\}^*\text{r})$.\\
    Suppose lr occurs in a decomposition $p_k(w_{i+1}\cdots w_j) \in (\text{l}\{\text{l},\text{r}\}^+\text{r})$ and corresponds to a substring $w'$. Then the same lr occurs in the part of the decomposition $p_k(w_1 \cdots w_j)$ that corresponds to the substring $w'$ of $w_{i+1}\cdots w_j$. Therefore,  $w' \in L \cup R$.
    If $w'$ is a suffix of $w_{i+1}\cdots w_j$, then lr is also a suffix of $p_k(w_1 \cdots w_j)$ and $w' \in R$.
        At the same time lr occurs in the part of the decomposition $p_{k}(w_{i+1}\cdots w_n)$ that corresponds to $w'$. 
        If $w'$ is a prefix of $w_{i+1}\cdots w_j$, then lr occurs as a prefix of $p_{k}(w_{i+1}\cdots w_n)$ and $w' \in L$.
        Therefore $p_{k+1}(w_{i+1}\cdots w_j) \in (\text{l}\{\text{l},\text{r}\}^*\text{r})$.\\
    \textbf{STEP 2:} The sequence $\{p_k(w_{i+1}\cdots w_j)\}_{k \geq 0}$ is finite, the last element $p_{\hat{k}}(w_{i+1}\cdots w_j) = \text{lr}$. \\
    We noticed in step 1 that every decomposition $p_k(w_{i+1}\cdots w_j)$ is a substring of $p_k(w_{i+1} \cdots w_n)$.
    Since $\{p_k(w_{i+1} \cdots w_n)\}_{k \geq 0}$ is finite due to Corollary~\ref{corrolary:pk}, $\{p_k(w_{i+1}\cdots w_j)\}_{k \geq 0}$ is also finite. 
    From part (ii) of step 1 it follows that the last element $p_{\hat{k}}(w_{i+1}\cdots w_j)$ equals lr. \\
    \textbf{STEP 3:} $w_{i+1}\cdots w_j \in L_{j-i} \cap R_{j-i}$. \\
    The process described in part (ii) of Step 1 shows that $p_{\hat{k}}(w_{1}\cdots w_j) \in (\text{l}\{\text{l},\text{r}\}^*\text{lr})$ and the last lr corresponds to $w_{i+1}\cdots w_j$, so $w_{i+1}\cdots w_j \in R_{j-i}$.
    Due to the same argument $w_{i+1}\cdots w_j$ corresponds to the first lr of $p_{\hat{k}}(w_{i+1}\cdots w_n) \in (\text{lr}\{\text{l},\text{r}\}^*\text{r})$, so $w_{i+1}\cdots w_j \in L_{j-i}$.
    $(L_{j-i}, R_{j-i})$ is a partition, so $L_{j-i} \cap R_{j-i} = \emptyset$ and $w_{i+1}\cdots w_j = \epsilon$, but this contradicts the assumption that $i < j$.
\end{proof}

\subsection{Characterization of maximal non-overlapping codes}
We will now show that $\mathcal{M}_{q,n}$ contains all maximal non-overlapping codes and determine which collections of partitions $(L_i, R_i)_{i=1,\dots,n-1}$ generate maximal non-overlapping codes.

\begin{proposition}
    Every maximal non-overlapping code is contained in $\mathcal{M}_{q,n}$.
\end{proposition}

\begin{proof}
    Let $X$ be a maximal $q$-ary $n$-letter non-overlapping code.
    Now construct the sets $(L_i,R_i)$ corresponding to $X$ for $i < n$ as follows
    \begin{align*}
      L_1 &\coloneqq \{x \in \Sigma \mid \text{there exists a word in $X$ beginning with $x$}\},   \\
      R_1 &\coloneqq \Sigma \setminus L_1,
    \end{align*}
    and for $1 < i < n$
    \begin{align*}
        L_i &\coloneqq \{x_1\cdots x_i \in \bigcup_{j < i} (L_j R_{i-j}) \mid \text{there exists a word in $X$ beginning with $x_1\cdots x_i$}\}, \\
        R_i &\coloneqq \bigcup_{j < i} (L_j R_{i-j}) \setminus L_i. 
    \end{align*}
    Note that if a word in $X$ ends in $x_1\cdots x_i \in \bigcup_{j < i} (L_j R_{i-j})$, then $x_1\cdots x_i \in R_i$ (otherwise $x_1\cdots x_i \in L_i$ and there exists a word in X that starts in $x_1 \cdots x_i$ by definition of $L_i$ which contradicts the fact that $X$ is non-overlapping).
    Define $Z \coloneqq \bigcup_{j < i} (L_j R_{i-j})$.
    We will prove that $X \subseteq Z$. Since $X$ is maximal and $Z$ is non-overlapping by Theorem~\ref{theorem:c6}, it immediately follows that $X = Z$.
    
    Let $w \in X$.
    We will show that $\{p_k(w)\}_{k\geq 0}$ is well-defined.
    Every letter in $w$ belongs to $L_1 \cup R_1 = \Sigma$ since $w\in \Sigma^n$. The first letter of $w$ belongs to $L_1$ by definition of $L_1$ and we explained earlier that the last letter belongs to $R_1$. Decomposition $p_0(w)$ is therefore well-defined.
    Now suppose $p_k(w) \in \{\text{l}\{\text{l},\text{r}\}^*\text{r}\}$ for some $k \geq 0$. If $p_k(w) = \text{lr}$, then there exists some $i$ such that $w \in (L_iR_{n-i})$ and $w \in Z$.
    Otherwise $p_k(w) \in \{\text{l}\{\text{l},\text{r}\}^+\text{r}\}$ and there exists some lr in $p_k(w)$ that corresponds to a proper substring $w'$ of $w$.
    The length of $w'$ is at most $n-1$ and $w' \in L\cup R$, so lr can be replaced by either l or r in $p_{k+1}(w)$.
    If $p_k(w)$ starts with lr that corresponds to a $k$-letters long substring of $w$, $w_L$, then by definition of $L_k$ $w_L \in L_k$ and lr is replaced by l in $p_{k+1}(w)$.
    If $p_k(w)$ ends with lr that corresponds to a $k$-letters long substring of $w$ $w_R$, then $w_R\not\in L_k$ since $X$ is a non-overlapping code. So $w_R \in R_k$ and lr is replaced by r in $p_{k+1}(w)$.
    The length of $p_{k+1}(w)$ is strictly smaller than the length of $p_k(w)$ so the sequence $\{p_k(w)\}_{k \geq 0}$ is indeed finite and its last element is lr.
    As shown earlier this implies that $w \in Z$.
\end{proof}

\begin{theorem}\label{theorem:characterization}
    $C \in \mathcal{M}_{q,n}$ is maximal if and only if all the following statements hold. \\[1em]
        (i) If $L_i$ is non-empty and $R_{n-i} = \emptyset$, then every $x \in L_i$ is a prefix in some $L_j$ such that $i < j < n$ and $R_{n-j}$ is non-empty or $q=2$, $i = \frac{n}{2}$ and $\lvert L_i\rvert = 1$.\\[.5em]
        (ii) If $R_i$ is non-empty and $L_{n-i} = \emptyset$, then every $x \in R_i$ is a suffix in some $R_j$ such that $i < j < n$ and $L_{n-j}$ is non-empty or $q=2$, $i = \frac{n}{2}$ and $\lvert R_i\rvert = 1$. \\[.5em]
        (iii) If $q=2$ and $L_{\frac{n}{2}} = \{u\}$ such that $u$ is not a prefix in $C$, then for every non-empty $L_j$, $2 \leq j \leq \frac{n}{2}-2$, the word in $(L_jL_{\frac{n}{2}})$ is a prefix in $C$. Moreover, the word in $(L_1L_{\frac{n}{2}})$ is a prefix in $C$ or $L_{\frac{n}{2}-1}$ is empty and for every non-empty $L_j$, $1 \leq j \leq n/2 - 2$, the word in $(L_jL_1L_{\frac{n}{2}})$ is a prefix in C.\\[.5em]
        (iv) If $q=2$ and $R_{\frac{n}{2}} = \{u\}$ such that $u$ is not a suffix in $C$, then for every non-empty $R_j$, $2 \leq j \leq \frac{n}{2}-2$, the word in $(R_{\frac{n}{2}}R_j)$ is a suffix in $C$. Moreover, the word in $(R_{\frac{n}{2}}R_1)$ is a suffix in $C$ or $R_{\frac{n}{2}-1}$ is empty and for every non-empty $R_j$, $1 \leq j \leq n/2 - 2$, the word in $(R_{\frac{n}{2}}R_1R_j)$ is a suffix in C.
\end{theorem}

\begin{proof}
    \textbf{STEP 1: If $C$ is maximal, then statements (i) to (iv) hold.}\\
    Suppose $C = \bigcup_{i<n} (L_i R_{n-i})$ is maximal.\\
    \textbf{(i)}
    Let $x \in L_i$. If for every $j$, $i < j < n$, $x$ is not a prefix in $L_j$ or $R_{n-j}$ is empty, then $x$ is not a prefix in 
    $\bigcup_{j = i}^{n-1} (L_j R_{n-j}) \subseteq C$.
    It is also not a prefix in $\bigcup_{j < i} (L_j R_{n-j})$, otherwise $R_{n-j}$ has a prefix that is a suffix of $x \in L_i$.
    So $x$ is not a prefix in $C$.
    It is also not a suffix in $C$ due to Proposition~\ref{proposition:z}.
    
    \textbf{CASE 1: $i \neq \frac{n}{2}$} \\
    $R_{n-i}$ is empty, therefore $L_{n-i} = \bigcup_{j < n-i} (L_j R_{n-i-j}) \setminus R_{n-i} = \bigcup_{j < n-i} (L_j R_{n-i-j})$ is non-empty. 
    Observe the set $(L_{n-i}x)$. If $yx$ is a suffix in $(L_{n-i}x)$, then $y$ is a suffix in $L_{n-i}$, and it is neither a prefix in $C$ nor in $L_{n-i}$ by Proposition~\ref{proposition:z}.
    Since $n - i \neq \frac{n}{2}$, then $x$ itself is not a prefix in $(L_{n-i}x)$, and we already showed that it is not a prefix in $C$.
    If $y$ is a suffix of length $k$ in $(L_{n-i}x)$ that is shorter than $x$, then it is not a prefix in $C$ by Proposition~\ref{proposition:z}. If $k \leq n-i$ then $y$ is not a prefix in $(L_{n-i}x)$ as it cannot be a prefix in $L_{n-i}$ by Proposition~\ref{proposition:z}. If $k > n -i$ then $y$ can be a prefix in $(L_{n-i}x)$ only if it has a suffix that is a prefix in $x$, but $x$ is itself non-overlapping. 
    Therefore $C \cup (L_{n-i}x)$ is a non-overlapping code. This contradicts the maximality of $C$.
    
    \textbf{CASE 2: $i = \frac{n}{2}$} \\
    Suppose there exists $y \in L_{\frac{n}{2}} \setminus \{x\}$.
    Observe the word $yx$.
    If $z$ is a suffix of $x$, then it is not a prefix of $y$ nor of any word in $C$ by Proposition~\ref{proposition:z}.
    If $z$ is a suffix of $yx$ longer than $\frac{n}{2}$, then $z$ is not a prefix of $yx$ nor of any word in $C$, otherwise $y$ is not non-overlapping or it has a suffix that is a prefix of a word in $C$ which contradicts Proposition~\ref{proposition:z}.
    So $C \cup \{yx\}$ is a non-overlapping code, but this contradicts the maximality of $C$.
    Now notice that for $j > 2$, $\lvert L_j \cup R_j \rvert = \sum_{k < j} \lvert L_k \rvert \lvert R_{j-k} \rvert \geq \lvert L_1 \rvert \lvert R_{j-1} \rvert + \lvert L_{j-1} \rvert \lvert R_1 \rvert  \geq  \lvert L_{j-1}\rvert + \lvert R_{j-1}\rvert = \lvert L_{j-1} \cup R_{j-1}\rvert$.
    If $j = 2$, then $\lvert L_j \cup R_j \rvert = \lvert L_1 \rvert  \lvert R_1 \rvert$.
    In both cases $\lvert L_\frac{n}{2} \rvert = 1$ is only possible if $\lvert L_1 \cup R_1 \rvert = 2$.\\[.5em]
    \textbf{(ii)} The proof is symmetric to the proof of statement (i).\\[.5em]
    \textbf{(iii)} The case $\lvert L_{\frac{n}{2}}\rvert$ is only possible if either $L_2 = \cdots = L_{\frac{n}{2}-2} = \emptyset$ or $R_2 = \cdots = R_{\frac{n}{2}-2} = \emptyset$. Suppose that for some $k \in \{1,\dots,\frac{n}{2}-2\}$ the word in $(L_kL_\frac{n}{2})$ is not a prefix in X. Then if $L_{\frac{n}{2}-k}$ is non-empty, $(L_{\frac{n}{2}-k}L_kL_{\frac{n}{2}}) \cup C$ is a non-overlapping code larger than C. If $k > 1$, $L_k$ is non-empty if and only if $L_{\frac{n}{2}-k}$ is non-empty. If $k=1$, $L_{\frac{n}{2} -1} = \emptyset$ and $L_l$ non-empty for some $l < \frac{n}{2}-1$ it follows that $(L_{\frac{n}{2}-l-1}L_lL_1L_{\frac{n}{2}}) \cup C$ is a non-overlapping code larger than C if the word in $(L_lL_1L_{\frac{n}{2}})$ is not a prefix in C.\\[.5em]
    \textbf{(iv)} The proof is symmetric to the proof of statement (iii).\\[1em]
    \textbf{STEP 2: If statements (i) to (iv) hold, then $C$ is maximal.}\\
    Let $C = \bigcup_{i<n} (L_i R_{n-i}) \in \mathcal{M}_{q,n}$ that satisfies (i) and (ii).
    Suppose $C$ is not maximal. Then there exists a word $w \in \Sigma^n \setminus C$ such that no proper prefix of $w$ is a suffix in $C$ and no proper suffix of $w$ is a prefix in $C$.
    Let $\{\hat{p}_k(w)\}_{k \geq 0}$ be a sequence of decompositions of $w \in \Sigma^n \setminus C$ in $(L\cup R)^+$ encoded with a binary word $\hat{p}_k$ over the alphabet $\{\text{l},\text{r}\}$ such that \\
        (a) $\hat{p}_0(w) \in \{\text{l},\text{r}\}^n$ with $\hat{p}_0(w)_i = \text{l}$ if $w_o \in L_1$ and $\hat{p}_0(w) = \text{r}$ if $w_i \in R_1$, \\
        (b) if $\hat{p}_k(w)$ contains lr as a proper substring, then obtain $\hat{p}_{k+1}(w) \in \{\text{l},\text{r}\}^+$ by replacing each occurrence of lr in $\hat{p}_k(w)$ by l if lr corresponds to a subword $w'$ of $w$ in $(L_iR_j) \subseteq L_{i+j}$ or by r if it corresponds to a subword $w'$ of $w$ in $(L_iR_j) \subseteq R_{i+j}$ for some positive integers $i$ and $j$.

    Decomposition $\hat{p}_0(w)$ is well-defined as every letter of $w$ is an element of $L_1 \cup R_1 = \Sigma$.
    If $\hat{p}_k(w)$ contains a substring lr that corresponds to a substring $w'$ of $w$ then there exist integers $i$ and $j$ such that $w' \in L_{i+j} \cup R_{i+j}$.
    The sequence is therefore well-defined.
    The length of decompositions is strictly decreasing, so the sequence $\{\hat{p}_k(w)\}_{k \geq 0}$ is finite.
    Denote the last element with $\hat{p}_{\hat{k}}(w)$.
    The sequences in $\{\text{l},\text{r}\}^+$ that have no lr as a proper substring are of the forms lr, r$^+$l$^+$, l$^+$, r$^+$.
    Decomposition $\hat{p}_{\hat{k}}(w)$ cannot be lr, otherwise $w \in C$.

    If $n$ is odd or every word in $L_{\frac{n}{2}}$ is a prefix in $C$ and every word in $R_{\frac{n}{2}}$ is a suffix in $C$, it follows from (i) that for every $x \in L$ there exists a word $w'\in C$ that starts in $x$ and from (ii) that for every $x \in R$ there exists a word $w' \in C$ that ends in $x$. Therefore $C \cup \{w\}$ is not non-overlapping if $\hat{p}_{\hat{k}}(w) \in \text{r}^+\text{l}^+ \cup \text{r}^+ \cup \text{l}^+$.
    
    Otherwise $q = 2$ and $\lvert L_\frac{n}{2} \cup R_\frac{n}{2} \rvert = 1$. The latter is only possible if either $L_2=\cdots=L_{\frac{n}{2}-2} = \emptyset$ or $R_2=\cdots=R_{\frac{n}{2}-2} = \emptyset$. Since every word in $L$ of length distinct from $\frac{n}{2}$ is a prefix in $C$ by (i) and every word in $R$ of length distinct from $\frac{n}{2}$ is a suffix in $C$ by (ii), $\hat{p}_{\hat{k}}(w) \in  \text{r}^+\text{l}^+$ implies $\hat{p}_{\hat{k}}(w) = \text{rl}$ and $w \in (R_{\frac{n}{2}}L_{\frac{n}{2}}) = \emptyset$. Therefore $\hat{p}_{\hat{k}}(w) \in \text{r}^+ \cup \text{l}^+$ and $w \in (L^+LL_{\frac{n}{2}}) \cup (R_{\frac{n}{2}}RR^+)$ since $w$ itself is non-overlapping.
    Suppose $R_{\frac{n}{2}}$ is empty (the proof when $L_{\frac{n}{2}}$ is empty is symmetric using statement (iv) instead of (ii)). The set $(R_\frac{n}{2}RR^+)$ is empty and hence $w \in (L^+LL_{\frac{n}{2}})$.
    If $(L_kL_\frac{n}{2})$ is non-empty for some $k \in \{2,\dots,\frac{n}{2}-2\}$, then $w \in (L_kL_\frac{n}{2})$ is a prefix in $C$ by (iii). Also the word in $(L_1L_{\frac{n}{2}})$ is a prefix in $C$ or $L_{\frac{n}{2}-1}$ is empty and for every nonempty $L_k$, $k \in \{1,\dots,\frac{n}{2}-2\}$, the word in $(L_kL_1L_\frac{n}{2})$ is a prefix in $C$. Hence $w \in (L_1L_{\frac{n}{2}-1}L_\frac{n}{2})$ and $L_{\frac{n}{2}-1}$ is non-empty, but the word in $(L_{\frac{n}{2}-1}L_\frac{n}{2})$ is a prefix in $(L_{\frac{n}{2}-1}R_{\frac{n}{2}+1}) \subseteq C$ since $L_{\frac{n}{2}}$ is not a prefix in C.
    No word can therefore be added to $C$, so $C$ is maximal.
\end{proof}

We showed earlier that the mapping from collections $(L_i, R_i)_{i=1,\dots,n-1}$ to non-overlapping codes is not injective. Nevertheless, we will show that almost every maximal non-overlapping code corresponds to exactly one collection of partitions.
Unfortunately, a characterization of maximal non-overlapping codes in terms of partition sizes cannot be given as demonstrated by Example~\ref{example:n6q3}, and we will therefore not use it directly when computing $S(q,n)$.

\begin{proposition}\label{proposition:maximal}
Let $C = \bigcup_{i < n} (L_iR_{n-i}) = \bigcup_{i < n} (\hat{L}_i\hat{R}_{n-i}) \in \mathcal{M}_{q,n}$ be a maximal non-overlapping code.
Then 
\begin{align*}
\forall i < \frac{n}{2}:\; L_i &= \hat{L}_i \text{ and } R_i = \hat{R}_i.\\
\end{align*}
Moreover, exactly one of the following statements holds.\\
    (i) $\forall i \geq \frac{n}{2}:\; L_i = \hat{L}_i$ \text{ and } $R_i = \hat{R}_i$, \\
    (ii) $q = 2$, $L_{\frac{n}{2}} = \hat{R}_\frac{n}{2}$ and $\hat{L}_{\frac{n}{2}} = R_\frac{n}{2} = \emptyset$, \\
    (iii) $q = 2$, $L_{\frac{n}{2}} = \hat{R}_\frac{n}{2} = \emptyset$ and $\hat{L}_{\frac{n}{2}} = R_\frac{n}{2}$.
\end{proposition}

\begin{proof}
Suppose there exists some $i$ such that $L_i \neq \hat{L}_i$ and $\forall j < i: L_j = \hat{L}_j$.
$L_j \cup R_j = \hat{L}_j \cup \hat{R}_j$ for all $j \leq i$ by definition, so $R_j = \hat{R}_j$ for all  $j < i$.
Furthermore at least one of the sets $L_i \cap \hat{R}_i$ and $R_i \cap \hat{L}_i$ is non-empty.
Without loss of generality suppose $L_i \cap \hat{R}_i \neq \emptyset$ (otherwise $R_i \cap \hat{L}_i$ is non-empty and a symmetric argument follows).
Let $x \in L_i \cap \hat{R}_i$. From Proposition~\ref{proposition:z} we know that $x$ is neither a prefix nor a suffix in $C$.
If $i \neq \frac{n}{2}$ then Theorem~\ref{theorem:characterization} implies that $C$ is not maximal, so $L_j = R_j$ for every $j$, $1 \leq j < n$.
If $i = \frac{n}{2}$ then Theorem~\ref{theorem:characterization} implies $q = 2$, $\lvert L_\frac{n}{2}\rvert = 1 = \lvert \hat{R}_\frac{n}{2} \rvert$ and $\hat{L}_\frac{n}{2} = R_\frac{n}{2} = \emptyset$, so  $L_\frac{n}{2} = \hat{R}_\frac{n}{2}$.
\end{proof}

\begin{corollary}\label{corollary:q>2}
    If $n$ is odd or $q \geq 3$, then every maximal non-overlapping code corresponds to exactly one partition $(L_i, R_i)_{i=1,\dots,n-1}$.
\end{corollary}

\begin{proof}
    The second and third case of Proposition~\ref{proposition:maximal} cannot hold for an odd $n$ nor for $q \geq 3$.
\end{proof}

\begin{example}
    \label{example:n6q3}
    Let $n=6$ and $\Sigma = \{0,1,2\}$.
    Set $L_1 = \{0\}$, $R_1 = \{1,2\}$, $L_2 = \emptyset$, $R_2 = \{01,02\}$, $L_3 = \{001\}$, $R_3 = \{002\}$, $L_4 = \{0002\}$, $R_4 = \{0011,0012\}$.
    If we set $L_5 = \{00011,00012\}$ and $R_5 = \{00101,00102,00021,00022\}$, we obtain 
    \begin{align*}
      C_1 =\{&000101,000102,000021, 000022,001002,000201,000202,000111,000112, \\
      &000121,000122  \},
    \end{align*}
     which is not maximal as $C_1 \cup \{001101\}$ is a non-overlapping code.
    Now take $L_5 = \{00101,00102\}$ and $R_5 = \{00011,00012,00021,00022\}$. Code
    \begin{align*}
        C_2 = \{&000011,00012,000021,000022,001002,000201,000202,001011,001012, \\
        &001021,001022\}
    \end{align*}
    has partitions of the same sizes as $C_1$ but is maximal due to Theorem~\ref{theorem:characterization}.
\end{example}

\section{An integer optimization problem}\label{sec:optimization}
Now that we know that $M_{q,n}$ contains all maximal non-overlapping codes and we have a formula to compute their sizes, we can clearly determine the maximum non-overlapping codes by maximizing the value $\sum_{i=1}^{n-1} \lvert L_i \rvert \lvert R_{n-i} \rvert$ over all constraints given by Construction~\ref{construction:fimmel}.
An integer optimization problem is formulated as follows (see Proposition~\ref{cor:basicoptproblem} below). We call this formulation SQN(q, n) throughout the paper.

\begin{proposition}[SQN]
\label{cor:basicoptproblem}
    
    \begin{align*}
    S(q,n) = & \max \sum_{i=1}^{n-1}  x_i y_{n-i} \\
        \textnormal{subject to }\quad & x_1 +  y_1  = q \\
        & x_i  +  y_i  = \sum_{j=1}^{i-1}  x_j  y_{i-j} \quad \forall i > 1 \\
        &  x_1 ,  y_1  > 0 \\
        & x_i,  y_i  \geq 0 \quad \forall i > 1 \\
        & x_i,  y_i  \in \mathbb{Z} \quad \forall i \geq 1.
    \end{align*}
    If $(x^*,y^*)$ is an optimal solution to the above optimization problem, then $C = \bigcup_{i=1}^{n-1} \left(L_iR_{n-i}\right)$ is a maximum non-overlapping code if it satisfies\\
        (i) $\left(L_1,R_1\right)$ is a partition of $\Sigma$ with $\lvert L_1 \rvert = x_1^*$ and \\
        (ii) for $n > i > 1: \; \left(L_i,R_i\right)$ is a partition of $C = \bigcup_{j=1}^{i-1} \left(L_jR_{i-j}\right)$ with $\lvert L_i \rvert = x_i^*$. \\
\end{proposition}
\begin{proof}
    The proposition follows directly from the definition of $S(q,n)$, Construction~\ref{construction:fimmel}, Theorem~\ref{theorem:c6} and Theorem~\ref{prop:size}
    if we denote the size of the set $L_i$ by $x_i$ and the size of the set $R_i$ by $y_i$.
\end{proof}

The evaluation of the objective function for all feasible solutions to SQN requires $O\left(q^{n^2}\right)$ time and
$\Theta\left(n\right)$ space. If we also want to store all the optimal solutions, space requirement increases to $\Theta\left(n m\right)$ where $m$ denotes the number of optimal solutions.
To determine $N(q,n)$ from the optimal solutions of SQN(q, n), we have to evaluate whether any pair of solutions of \mbox{SQN(q, n)} constructs the same maximum code to prevent double counting. Proposition~\ref{proposition:nqn} shows that for most parameter values one can compute the value $N(q,n)$ directly.

\begin{proposition}\label{proposition:nqn}
     \[N(q, n) = \sum_{\substack{x,y \textnormal{ optimal}\\ \textnormal{solution of SQN(q,n)} }} \binom{q}{x_1}\cdot \prod_{i=2}^{n-1}\binom{\sum_{j=1}^{i-1} x_j y_{i-j} }{x_i} \]
if at least one of the following holds: \\
     (i) $q > 2$,  \\
     (ii) $n$ is odd, \\
     (iii) no pair of solutions $(x,y), (\hat{x}, \hat{y})$ of SQN(q,n) satisfies $x_i = \hat{x}_i$ for all $i < \frac{n}{2}$, $x_\frac{n}{2} = \hat{y}_\frac{n}{2} = 0$, $\hat{x}_\frac{n}{2} = y_\frac{n}{2} = 1$.
\end{proposition}
\begin{proof}
Let $(x,y)$ be an optimal solution of SQN.
There are $\binom{q}{x_1}$ choices for a partition of $\Sigma$ into $\left(L_1,R_1\right)$ such that $\lvert L_1\rvert = x_1$ and $\lvert R_1\rvert = y_1$.
For fixed partitions $\left(L_1,R_1\right), \dots, \left(L_{i-1}, R_{i-1}\right)$ there are $\binom{\sum_{j=1}^{i-1}\lvert L_j\rvert\lvert R_{i-j}\rvert}{x_i} = \binom{\sum_{j=1}^{i-1} x_j y_{i-j} }{x_i}$ choices for a partition of $\bigcup_{j=1}^{i-1} \left(L_j R_{i-j}\right)$ with $\lvert L_i \rvert = x_i$ and $\lvert R_i \rvert = y_i$.
No code is double counted as the non-overlapping codes corresponding to solutions of SQN(q, n) are distinct due to Proposition~\ref{proposition:maximal} and Corollary~\ref{corollary:q>2}.
\end{proof}

\subsection{Reduction of SQN}\label{sec:conditions}
The set of feasible solutions of SQN can be reduced.
We provide some equivalence transformations on $(x,y)$ that preserve the value of the objective function.
These enable us to evaluate only one feasible solution for each equivalence set.
Moreover, we will show that every maximum $C\in \mathcal{M}_{q,n}$ either satisfies the property 
\begin{equation}\label{eq:property}
\lvert L_i\rvert \lvert R_i\rvert = 0 \qquad \forall i > \frac{n}{2},
\end{equation}
or it can be mapped to a non-overlapping code that satisfies property~\eqref{eq:property} using a transformation we will define later.

\begin{proposition}\label{proposition:symmetry}
    $(x,y)$ is an optimal solution of SQN if and only if $(y,x)$ is optimal for SQN.
\end{proposition}
\begin{proof}
    Feasible solutions $(x,y)$ and $(y,x)$ have the same value of the objective functions $\sum_{i=1}^{n-1}x_iy_{n-i} = \sum_{i=1}^{n-1}y_ix_{n-i}$.
\end{proof}

\begin{proposition}\label{proposition:evenq}
    Let $((x_1,\dots, x_{n-1}), (y_1,\dots, y_{n-1}))$ be a feasible solution of SQN.
    If there exists some $i > 1$ such that $x_k = y_k$ for all $k \leq i$, then
    $((x_1,\dots, x_{n-1}), (y_1,\dots, y_{n-1}))$ is an optimal solution of SQN
    if and only if $((x_1,\dots,x_i,y_{i+1},\dots, y_{n-1}), (y_1,\dots,y_i,x_{i+1},\dots x_{n-1}))$ is an optimal solution.
\end{proposition}
\begin{proof}
    Clearly $((x_1,\dots,x_i,y_{i+1},\dots, y_{n-1}), (y_1,\dots,y_i,x_{i+1},\dots, x_{n-1}))$ is feasible.
    Let $o$ denote the objective function of $((x_1,\dots, x_{n-1}), (y_1,\dots, y_{n-1}))$ and $\hat{o}$ the objective function of $((x_1,\dots,x_i,y_{i+1},\dots, y_{n-1}), (y_1,\dots,y_i,x_{i+1},\dots, x_{n-1}))$.

    \begin{align*}
    \hat{o} = &\sum_{\substack{1 \leq j \leq i \\ n - j > i}}x_jx_{n-j} +
    \sum_{\substack{1 \leq j \leq i \\ n - j \leq i}} x_j y_{n-j}
    + \sum_{\substack{i < j < n \\ n - j > i}} y_j x_{n-j} + 
    \sum_{\substack{i < j < n \\ n - j \leq i}} y_jy_{n-j} \\
    = &\sum_{\substack{1 \leq j \leq i \\ n - j > i}}y_jx_{n-j} +
    \sum_{\substack{1 \leq j \leq i \\ n - j \leq i}} y_j x_{n-j}
    + \sum_{\substack{i < j < n \\ n - j > i}} y_j x_{n-j} + 
    \sum_{\substack{i < j < n \\ n - j \leq i}} y_jx_{n-j} \\
    = & \sum_{1 \leq j \leq i} y_jx_{n-j} + \sum_{i < j < n} y_jx_{n-j} = o. \mbox{\tag*{\qedhere}}
    \end{align*}
\end{proof}

Before providing a formula for expressing the sizes of the sets
$ L_{i+1}, R_{i+1}, \dots , L_{n-1}, R_{n-1}$
is terms of $\lvert L_1\rvert, \lvert R_1\rvert , \dots , \lvert L_i\rvert , \lvert R_i\rvert $ when property~\eqref{eq:property} holds, let us first introduce few more definitions.

\begin{definition}\label{definition:coefficients}
Let $i$ be an integer such that $n > i > \frac{n}{2}$ and $\lvert L_{k}\rvert \lvert R_{k}\rvert  = 0$ for all $k > i$.\\
    (i) Coefficients $\delta_{R,k}$ and $\delta_{L,k}$ for $k > i$ express which of the sets $R_k$ and $L_k$ is non-empty:
    \begin{align*}
        \delta_{R,k} &\coloneqq \begin{cases}
                     1 & \text{if } L_k=\emptyset \\ 0 & \text{if } R_k=\emptyset,
                \end{cases} \\
         \delta_{L,k} &\coloneqq \begin{cases}
                     1 & \text{if } R_k=\emptyset \\ 0 & \text{if } L_k=\emptyset.
\end{cases}       
    \end{align*} \\
(ii) The coefficients $p_{jk}$ satisfy the following relations:
\begin{align*}
    p_{jj} &\coloneqq 1, \\
    p_{jk} &\coloneqq \sum_{m=k}^{j-1} \left(\delta_{L,i+m} \lvert R_{j-m}\rvert  + \delta_{R,i+m}\lvert L_{j-m}\rvert \right)p_{mk}.
\end{align*}\\
(iii)    Condition $c_i$ is defined as follows. We will use its value to determine which part of partition $(L_i, R_i)$ is empty. 
\begin{align*}
        c_i \coloneqq  \sum_{j=1}^{n-1-i} \left(\delta_{L,i+j}\lvert R_{n-i-j}\rvert  + \delta_{R,i+j}\lvert L_{n-i-j}\rvert \right) \left(\sum_{l=1}^j p_{jl}(\lvert R_l\rvert  - \lvert L_l\rvert )\right)
        +\lvert R_{n-i}\rvert \\  - \lvert L_{n-i}\rvert.
    \end{align*}    
\end{definition}

\begin{proposition}\label{lemma:upper_half_size}
Let $i$ be an integer such that $n > i > \frac{n}{2}$ and let $\lvert L_{i+j}\rvert \lvert R_{i+j}\rvert  = 0$ for all $1 \leq j < n - i$.
Then \[ \lvert L_{i+j}\rvert  = \delta_{L, i+j} \sum_{l=1}^j\sum_{k=l}^i p_{jl} \lvert L_k\rvert  \lvert R_{i+l-k}\rvert  \]
\[\lvert R_{i+j}\rvert  = \delta_{R, i+j} \sum_{l=1}^j\sum_{k=l}^i p_{jl} \lvert L_k\rvert  \lvert R_{i+l-k}\rvert . \]
\end{proposition}

\begin{proof}
    We will prove the proposition by induction on $j$.
    Let $j=1$.
    Because $\lvert L_{i+1} \rvert \lvert R_{i+1}\rvert  = 0,$
    \begin{align*}
        \lvert L_{i+1}\rvert  &= \delta_{L,i+1} \sum_{k=1}^{i} \lvert L_{k}\rvert \lvert R_{i+1-k}\rvert  \\
        &=\delta_{L,i+1} \sum_{k=1}^{i}  p_{11} \lvert L_{k}\rvert \lvert R_{i+1-k}\rvert  \\
        & = \delta_{L,i+1} \sum_{l=1}^1 \sum_{k=1}^{i} p_{1l} \lvert L_{k}\rvert \lvert R_{i+1-k}\rvert ,
    \end{align*}
    and following the same procedure $\lvert R_{i+1}\rvert  = \delta_{R,i+1} \sum_{l=1}^1 \sum_{k=1}^{n} p_{1l} \lvert L_{k}\rvert \lvert R_{i+1-k}\rvert $.
    Now let us suppose that the statement holds for all positive integers smaller than $j$.
    The set $\bigcup_{k=1}^{i+j-1} \left(L_k R_{i+j-k}\right)$ can be expressed as a union of those words that
    end in $R_{i+k}$,
    words that start in $L_{i+k}$,  and
    words that start and end in substrings shorter or equal to $i$.
    Hence
    \begin{align*}
        \bigcup_{k=1}^{i+j-1} \left(L_k R_{i+j-k}\right) = &\bigcup_{k=i+1}^{i+j-1} \left(L_{i+j-k} R_k\right) \cup \bigcup_{k=i+1}^{i+j-1} \left(L_k R_{i+j-k}\right) \\
        &\cup \bigcup_{k=j}^i \left(L_{k}R_{i+j-k}\right).
    \end{align*}
    Since $i > \frac{n}{2}$, these three sets never overlap, so the size of their union is equal to the sum of their sizes.
    The size of $\bigcup_{k=j}^i \left(L_{k}R_{i+j-k}\right)$ can be expressed using the coefficients $p_{jj}$ as follows.
    \begin{align*}
        \lvert \bigcup_{k=j}^i \left(L_{k}R_{i+j-k}\right) \rvert  = \sum_{k=j}^i \lvert L_{k}\rvert \lvert R_{i+j-k}\rvert  = \sum_{k=j}^i p_{jj} \lvert L_{k}\rvert \lvert R_{i+j-k}\rvert.
    \end{align*}
    To compute the size of the sets $\bigcup_{k=i+1}^{i+j-1} \left(R_k L_{i+j-k}\right) \cup \bigcup_{k=i+1}^{i+j-1} \left(L_k R_{i+j-k}\right)$,
    we apply the induction hypothesis, then change the order of summation, introduce $\bar{k} \coloneqq k - i$, and recognise the formula for $p_{jl}$.
    \begin{align*}
        &\sum_{k=i+1}^{i+j-1} \lvert R_k\rvert  \lvert L_{i+j-k}\rvert  + \sum_{k=i+1}^{i+j-1} \lvert L_k\rvert  \lvert R_{i+j-k}\rvert  \\
        &= \sum_{k=i+1}^{i+j-1} \sum_{l=1}^{k-i} \sum_{m=l}^i p_{k-i,l} \lvert L_m\rvert \lvert R_{i+l-m}\rvert \left(\lvert L_{i+j-k}\rvert  \delta_{R,k}  + \lvert R_{i+j-k}\rvert   \delta_{L,k}\right) \\
        &= \sum_{l=1}^{j-1}\sum_{m=l}^{i}\sum_{k=i+l}^{i+j-1}p_{k-i,l} \lvert L_m\rvert \lvert R_{i+l-m}\rvert \left(\lvert L_{i+j-k}\rvert  \delta_{R,k}  + \lvert R_{i+j-k}\rvert   \delta_{L,k}\right) \\
        &= \sum_{l=1}^{j-1}\sum_{m=l}^{i} \left(\sum_{\bar{k}=l}^{j-1}p_{\bar{k},l} \left(\lvert L_{j-\bar{k}}\rvert  \delta_{R,i+\bar{k}}  + \lvert R_{j-\bar{k}}\rvert   \delta_{L,i+\bar{k}}\right)\right)\lvert L_m\rvert \lvert R_{i+l-m}\rvert  \\
        &= \sum_{l=1}^{j-1}\sum_{m=l}^{i} p_{jl} \lvert L_m\rvert \lvert R_{i+l-m}\rvert.
    \end{align*}
    We now combine the results to obtain
    \begin{align*}
        \lvert L_{i+j}\rvert  &= \delta_{L,i+j} \sum_{k=1}^{i+j-1}\lvert L_{k}\rvert \lvert R_{i+j-k}\rvert  \\
        &= \delta_{L,i+j} \left(\sum_{k=j}^i p_{jj} \lvert L_k\rvert  \lvert R_{i+j-k}\rvert   + \sum_{l=1}^{j-1}\sum_{m=l}^i p_{jl} \lvert L_m\rvert  \lvert R_{i+l-m}\rvert  \right) \\
        &= \delta_{L,i+j} \sum_{l=1}^{j}\sum_{m=l}^i p_{jl} \lvert L_m\rvert  \lvert R_{i+l-m}\rvert ,
    \end{align*}
    and equivalently
    \begin{equation*}
        \lvert R_{i+j}\rvert
        = \delta_{R,i+j} \sum_{l=1}^{j}\sum_{m=l}^i p_{jl} \lvert L_m\rvert  \lvert R_{i+l-m}\rvert.  \mbox{\tag*{\qedhere}}
    \end{equation*}
\end{proof}

\begin{theorem}\label{theorem:maximum}
    There exists a maximum non-overlapping code $C = \bigcup_{i=1}^{n-1} \left(L_i R_{n-i} \right)$
    satisfying the property $\forall i > \frac{n}{2}: \; \vert L_i\rvert \lvert R_i\rvert = 0 $. \\
    Moreover, when $i > \frac{n}{2}$ the condition $c_i$ from Definition~\ref{definition:coefficients}  satisfies:\\
    (i) $c_i \geq 0$ if  $R_i = \emptyset$ and \\
    (ii) $c_i \leq 0$ if $L_i = \emptyset$ .
\end{theorem}

\begin{proof}
    Let $C = \bigcup_{i=1}^{n-1} \left(L_i R_{n-i} \right)$ be a maximum non-overlapping code satisfying the property that there exists an
    integer $j$, $\frac{n}{2} < j \leq n - 1$, such that $\lvert L_{j}\rvert \lvert R_{j}\rvert  \neq 0$.
    Let us denote the largest such integer by $m$.
    Now rearrange the formula $\lvert C\rvert  = \sum_{j=1}^{n-1} \lvert L_j \rvert \lvert R_{n-j} \rvert$ in order to apply Proposition~\ref{lemma:upper_half_size} to it.
    \begin{align*}
        \lvert C\rvert
        =&\sum_{j=m+1}^{n-1} \left(\lvert L_j\rvert \lvert R_{n-j}\rvert  + \lvert R_j\rvert \lvert L_{n-j}\rvert \right) + \sum_{j=n-m}^m \lvert L_j\rvert  \lvert R_{n-j}\rvert  \\
        =&\sum_{j=1}^{n-1-m} \left(\lvert L_{m+j}\rvert \lvert R_{n-m-j}\rvert  + \lvert R_{m+j}\rvert \lvert L_{n-m-j}\rvert \right)
        + \sum_{j=n-m}^m \lvert L_j\rvert  \lvert R_{n-j}\rvert  \\
        =& \sum_{j=1}^{n-1-m} \left( \delta_{L,m+j}\lvert R_{n-m-j}\rvert + \delta_{R,m+j}\lvert L_{n-m-j}\rvert \right) \left(\sum_{l=1}^j \sum_{k=l}^m p_{jl} \lvert L_k\rvert  \lvert R_{m+l-k}\rvert \right) \\
        &+ \sum_{j=n-m}^m \lvert L_j\rvert  \lvert R_{n-j}\rvert.
    \end{align*}
    Define 
    $\bar{C} \coloneqq \bigcup_{i=1}^{n-1} \left(\bar{L}_i \bar{R}_{n-i} \right)$ such that
    \[\forall i < m: \bar{L}_i\coloneqq L_i \text{ and } \bar{R}_i\coloneqq R_i,\]
    \[\bar{L}_{m}  \coloneqq \begin{cases}
        \emptyset & \text{if } c_m \leq 0 \\
        \bigcup_{i=1}^{m-1} \left(\bar{L}_i \bar{R}_{m-i}\right) & \text{if } c_m > 0, \\
    \end{cases}
    \]
    \[\bar{R}_{m}  \coloneqq \begin{cases}
    \bigcup_{i=1}^{m-1} \left(\bar{L}_i \bar{R}_{m-i}\right) & \text{if } c_m \leq 0 \\
        \emptyset & \text{if } c_m > 0 ,\\
    \end{cases}
    \]
    and $\forall j \in \{1, \dots, n-1-m\}$: 
    \[\bar{L}_{m+j} = \begin{cases}
        \emptyset & \text{if } L_{m+j} = \emptyset \\
        \bigcup_{i=0}^{m+j-1} (\bar{L}_i\bar{R}_{m+j-i}) & \text{if } R_{m+j} = \emptyset,
    \end{cases}\]
    \[\bar{R}_{m+j} = \begin{cases}
        \bigcup_{i=0}^{m+j-1} (\bar{L}_i\bar{R}_{m+j-i}) & \text{if } L_{m+j} = \emptyset \\
        \emptyset & \text{if } R_{m+j} = \emptyset.
    \end{cases}\]
    $\bar{C}$ is non-overlapping due to Theorem~\ref{theorem:c6}.
    Now we observe the size of the code $\bar{C}$. From Proposition~\ref{lemma:upper_half_size} we get 
    \[ \lvert \bar{L}_{m+j}\rvert  \coloneqq \delta_{L, m+j} \sum_{l=1}^j\sum_{k=l}^m p_{jl} \lvert \bar{L}_k\rvert  \lvert \bar{R}_{m+l-k}\rvert ,  \]
    \[\lvert \bar{R}_{m+j}\rvert  \coloneqq \delta_{R, m+j} \sum_{l=1}^j\sum_{k=l}^m p_{jl} \lvert \bar{L}_k\rvert  \lvert \bar{R}_{m+l-k}\rvert , \]
    and from Theorem~\ref{prop:size}
    \begin{align*}
        \lvert \bar{C}\rvert  =& \sum_{j=1}^{n-1} \lvert \bar{L}_j\rvert  \lvert \bar{R}_{n-j}\rvert
        =\sum_{j=m+1}^{n-1} \left(\lvert \bar{L}_j\rvert \lvert \bar{R}_{n-j}\rvert  + \lvert \bar{R}_j\rvert \lvert \bar{L}_{n-j}\rvert \right) + \sum_{j=n-m}^m \lvert \bar{L}_j\rvert  \lvert \bar{R}_{n-j}\rvert  \\
        =&\sum_{j=1}^{n-1-m} \left(\lvert \bar{L}_{m+j}\rvert \lvert \bar{R}_{n-m-j}\rvert  + \lvert \bar{R}_{m+j}\rvert \lvert \bar{L}_{n-m-j}\rvert \right)
        + \sum_{j=n-m}^m \lvert \bar{L}_j\rvert  \lvert \bar{R}_{n-j}\rvert  \\
        =& \sum_{j=1}^{n-1-m} \left(\delta_{L,m+j}\lvert \bar{R}_{n-m-j}\rvert + \delta_{R,m+j}\lvert \bar{L}_{n-m-j}\rvert\right)  \left(\sum_{l=1}^j \sum_{k=l}^m p_{jl} \lvert \bar{L}_k\rvert  \lvert \bar{R}_{m+l-k}\rvert \right) \\
        & + \sum_{j=n-m}^m \lvert \bar{L}_j\rvert  \lvert \bar{R}_{n-j}\rvert .
    \end{align*}
    We notice that in the above sum only $\bar{L}_i$ and $\bar{R}_i$ with $1 \leq i \leq m$ occur.
    We know that unless $i=m$, it holds that $\bar{L}_i = L_i$ and $\bar{R}_i = R_i$.
    Therefore
    \begin{align*}
        \lvert \bar{C}\rvert
        =& \sum_{j=1}^{n-1-m} \delta_{L,m+j}\lvert R_{n-m-j}\rvert \left(\sum_{l=1}^j p_{jl} \left(\lvert \bar{L}_m\rvert  - \lvert L_m\rvert \right) \lvert R_{l}\rvert + \sum_{l=1}^j  p_{jl} \lvert L_{l}\rvert \left(\lvert \bar{R}_m\rvert  - \lvert R_m\rvert \right)\right) \\
        +& \sum_{j=1}^{n-1-m} \delta_{R,m+j}\lvert L_{n-m-j}\rvert  \left(\sum_{l=1}^j  p_{jl} \left(\lvert \bar{L}_m\rvert  - \lvert L_m\rvert \right) \lvert R_{l}\rvert  + \sum_{l=1}^j  p_{jl} \lvert L_{l}\rvert \left(\lvert \bar{R}_m\rvert  - \lvert R_m\rvert \right) \right)\\
        +& \left(\lvert \bar{L}_m\rvert  - \lvert L_m\rvert \right) \lvert R_{n-m}\rvert  + \lvert L_{n-m}\rvert  \left(\lvert \bar{R}_m\rvert  - \lvert R_m\rvert \right) + \lvert C\rvert.
    \end{align*}
    We know that
    \[\lvert R_m \rvert = \sum_{i=1}^{m-1} \lvert L_i\rvert \lvert R_{m-1-i}\rvert  - \lvert L_m\rvert,\]
    \[\lvert \bar{R}_m \rvert = \sum_{i=1}^{m-1} \lvert L_i\rvert \lvert R_{m-1-i}\rvert  - \lvert \bar{L}_m\rvert,\]
    and therefore \[\lvert \bar{R}_m\rvert  - \lvert R_m\rvert  = \lvert L_m\rvert  - \lvert \bar{L}_m\rvert.\]
    \begin{align*}
        \lvert \bar{C}\rvert
        =& \left(\lvert \bar{L}_m\rvert  - \lvert L_m\rvert \right) \sum_{j=1}^{n-1-m} \left(\delta_{L,m+j}\lvert R_{n-m-j}\rvert + \delta_{R,m+j}\lvert L_{n-m-j}\rvert\right) \left(\sum_{l=1}^j p_{jl} \left(\lvert R_{l}\rvert  - \lvert L_{l}\rvert \right)\right) \\
        &+ \left(\lvert \bar{L}_m\rvert  - \lvert L_m\rvert \right) \left(\lvert R_{n-m}\rvert  - \lvert L_{n-m}\rvert \right) + \lvert C\rvert\\
        =& \left(\lvert \bar{L}_m\rvert  - \lvert L_m\rvert \right) c_m + \lvert C\rvert.
    \end{align*}
    Recall that we set \[\lvert\bar{L}_m \rvert = 0 < \lvert L_m \rvert \quad \text{ if } \; c_m  < 0,\] and \[\lvert \bar{L}_m \rvert = \sum_{i=1}^{m-1} \lvert L_i \rvert \lvert R_{m-1-i} \rvert > \lvert L_m \rvert \quad \text{if }  \; c_m > 0,\]
    so $\lvert \bar{C}\rvert  > \lvert C\rvert,$ as soon as $c_m \neq 0$ which contradicts the fact that $C$ is maximum.
    In the other case, when $c_m = 0$, we get $\lvert \bar{C}\rvert  = \lvert C\rvert $.
    Notice that in this case both $\lvert \bar{C}\rvert $ and $\lvert C\rvert $ are expressed as the same function
    in $\lvert L_1\rvert ,\lvert R_1\rvert , \dots, \lvert L_{m-1}\rvert , \lvert R_{m-1}\rvert $,
    and indeed $\lvert \bar{L}_j\rvert  \lvert \bar{R}_j\rvert  = 0$ for all $\frac{n}{2}<j\leq m-1$ if and only if
    $\lvert L_j\rvert  \lvert R_j\rvert  = 0$ for all $\frac{n}{2}<j\leq m-1$.
    The theorem therefore holds.
\end{proof}

The proof revealed that it is easy to determine all maximum non-overlapping codes if we know all solutions that satisfy Property~\eqref{eq:property}. 
The algorithm to generate them is summarized in the following corollary.

\begin{corollary}\label{cor:1}
    Let $C = \bigcup_{i=1}^{n-1} \left(L_i R_{n-i} \right)$ be a maximum non-overlapping code with $c_m = 0$ for some $m > \frac{n}{2}$.
    Define $\hat{L}_i \coloneqq L_i$ for $i < m$,
    $(\hat{L}_m, \hat{R}_m)$ a partition of $\bigcup_{i <m} (L_i R_{m-i})$,
    and for $i > m$
    \begin{align*}
        \hat{L}_i &\coloneqq \begin{cases}
        \bigcup_{j <i } (\hat{L}_j \hat{R}_{i-j}) & \text{if } c_i \geq 0 \\
        \emptyset & \text{if } c_i < 0, \end{cases} \\
        \hat{R}_i &\coloneqq \begin{cases}
         \emptyset & \text{if } c_i \geq 0 \\
        \bigcup_{j <i } (\hat{L}_j \hat{R}_{i-j}) & \text{if } c_i < 0. \end{cases} \\
    \end{align*}
    $C = \bigcup_{i=1}^{n-1} \left(\hat{L}_i \hat{R}_{n-i} \right)$ is a maximum non-overlapping code.
\end{corollary}

\section{Exact formulas for four-letter codes}\label{formula:4}
\begin{theorem}
Let $q \geq 3$.
The number of distinct maximal four-letter $q$-ary non-overlapping codes equals \\
 \begin{align*}
 \sum_{m=1}^{q-1} \binom{q}{m} \left( 2(2^{q-m} - 1)^{m(q-m)} + (2^m + 2^{q-m})^{m(q-m)} - 2^{m^2(q-m)} - 2^{m(q-m)^2} \right).
 \end{align*}
 
\end{theorem}
\begin{proof}
    By Corollary~\ref{corollary:q>2} it is sufficient to count the number of distinct partitions $(L_1, R_1), (L_2, R_2), (L_3, R_3)$ that satisfy the requirements of Theorem~\ref{theorem:characterization}.
    Sets $L_1$ and $R_1$ cannot be empty by definition.
    If $L_2$ is empty, then for every $y\in R_2$ there exists $x \in L_1$ such that $xy \in R_3$.
    $R_3$ is therefore non-empty, but $L_3$ can be empty as every $y \in R_1$ occurs as a suffix in $R_2 = (L_1R_1)$.
    Therefore there are
    \begin{align*}
        \sum_{m=1}^{q-1} \binom{q}{m} \prod_{j=1}^{m(q-m)} \left(\sum_{i=1}^{q-m} \binom{q-m}{i}  \right)  &=
        \sum_{m=1}^{q-1} \binom{q}{m} \prod_{j=1}^{m(q-m)} \left(2^{q-m} - 1\right)  \\
        &=\sum_{m=1}^{q-1} \binom{q}{m} \left(2^{q-m} - 1\right)^{m(q-m)}
    \end{align*}
    maximal non-overlapping codes with empty $L_2$. There are also as many maximal codes when $R_2$ is empty due to a symmetric observation.\\
    Suppose $L_2$ and $R_2$ are both non-empty. If $L_3$ and $R_3$ are non-empty, then the code is clearly maximal.
    If $L_3$ is empty, then every $x \in R_1$ is a suffix in $(L_2R_1) \subseteq R_3$, so the code is maximal.
    If $R_3$ is empty, then every $x \in L_1$ is a prefix in $(L_2R_1) \subseteq L_3$, so the code is maximal.
    Therefore there are
    \begin{align*}
        \sum_{m=1}^{q-1} \binom{q}{m}& \sum_{i=1}^{q(m-q)-1}  \binom{m(q-m)}{i} \sum_{j=0}^{m^2(q-m) + i(q-2m)} \binom{m^2(q-m) + i(q-2m)}{j} = \\
        &= \sum_{m=1}^{q-1} \binom{q}{m} \sum_{i=1}^{q(m-q)-1}  \binom{m(q-m)}{i} 2^{m^2(q-m) + i(q-2m)} \\
        &= \sum_{m=1}^{q-1} 2^{m^2(q-m)}\binom{q}{m} \sum_{i=1}^{q(m-q)-1}  \binom{m(q-m)}{i} 2^{(q-2m)i} \\
        &= \sum_{m=1}^{q-1} 2^{m^2(q-m)}\binom{q}{m} \left((1 + 2^{q-2m})^{m(q-m)} - 1 - 2^{m(q-m)(q-2m)}\right)
    \end{align*}
    maximal non-overlapping codes with $\lvert L_2\rvert \lvert R_2\rvert \neq 0$.
\end{proof}

Blackburn~\cite{Blackburn:2015} computed $S(q,2)$ and $S(q,3)$ by proving that Construction~\ref{construction:blackburn}
with $k = 1$ and $k = 2$ respectively gives us codes of maximum sizes.
Using the results from the previous section, we further prove that this construction with $k = 3$ always yields the optimum
for four-letter codes.
Before showing the result, we will provide Lemma~\ref{lemma:nhalves} and Proposition~\ref{proposition:blackburn_optimalsize}
regarding the sizes of codes generated by Construction~\ref{construction:blackburn}.

\begin{lemma}
    \label{lemma:nhalves}
    Let $q, n, m \leq \frac{n}{2}$ and $2 \leq j \leq n$ be positive integers, and
    $f$ the following function
    \begin{align*}
        f(j)= (n-m-1) \left(\frac{(n-1)q - m}{n}\right)^{n-j+1}\prod_{i=2}^{j-2}\frac{n-i}{i} \\+ \frac{n-q-m}{n}\sum_{i=0}^{n-j}\binom{n-1}{i}\left(\frac{(n-1)q-m}{n}\right)^i.
    \end{align*}
    Then $f(j) \geq f(j+1)$.
\end{lemma}
\begin{proof}
    The coefficient in front of $\left(\frac{(n-1)q-m}{n}\right)^{n-j}$ in $f(j)$ is
    $\left(\frac{(n-m-1)((n-1)q - m)}{n} + \frac{(n-q-m)(n-1)(n-j)}{nj}\right)\prod_{i=2}^{j-2}\frac{n-i}{i}$.
    Because $m \leq \frac{n}{2}$, it holds that $n - m -1 \geq \frac{n - 2}{2} \geq \frac{n - j}{j}$,
    so the coefficient in front of $\left(\frac{(n-1)q-m}{n}\right)^{n-j}$ in $f(j)$ is greater or equal to
    $\left(\frac{(n-1)q - m}{n} + \frac{(n-q-m)(n-1)}{n}\right)\prod_{i=2}^{j-1}\frac{n-i}{i}$.
    We now compute the sum
    \[\frac{(n-1)q - m}{n} + \frac{(n-q-m)(n-1)}{n} =  n - m - 1,\]
    and see that the statement follows.
\end{proof}

\begin{proposition}\label{proposition:blackburn_optimalsize}
Let $C$ be the largest $q$-ary non-overlapping code of length $n$ obtained by Construction~\ref{construction:blackburn} with $k = n - 1$.
Then the value of the parameter $l$ equals:
\begin{align*}
    l = \begin{cases}
        \frac{(n-1)q}{n} & \text{if $n$ divides $q$} \\
        \left\lceil{\frac{(n-1)q}{n}}\right\rceil & \text{if } (n-1)q  \mod{n} =  n - 1  \\
        \left\lfloor{\frac{(n-1)q}{n}}\right\rfloor & \text{if } (n - 1)q \mod{n} \leq \frac{n}{2}.
    \end{cases}
\end{align*}
\end{proposition}

\begin{remark}
    Note that we skip the case $\frac{n}{2} < (n - 1)q \mod{n} < n - 1 $, as we do not need it throughout our paper.
    Although one would wish that in that case $l = \left\lceil{\frac{(n-1)q}{n}}\right\rceil$, this does not always hold.
    If we take for example $n=11$ and $q=16$, then the code with parameter $l=14$ contains 576 650 390 625
    codewords, but the code with parameter $l=15$ contains 578 509 309 952 codewords.
\end{remark}

\begin{proof}
    The largest $q$-ary non-overlapping code of length $n$ obtained by Construction~\ref{construction:blackburn} with $k = n - 1$
    has size $\max_{l \in \{1,\dots, q-1\}} f(l),$ where $f(l) = (q - l) l^{n-1}$.
    First we observe the first derivative of $f$, $f'(l) = l^{n-2}((n-1)q - nl)$ and notice that $f(l)$ is strictly
    increasing on the interval $\left[1, \frac{(n-1)q}{n}\right)$ and strictly decreasing on the interval $\left(\frac{(n-1)q}{n}, q-1\right]$.
    If $n$ divides $q$, the maximum $\max_{l \in \{1,\dots, q-1\}} f(l)$ is achieved at $l=\frac{(n-1)q}{n}$.

    Now, let us suppose $q$ is not a multiplier of $n$.
    The maximum $\max_{l \in \{1,\dots, q-1\}} f(l)$ is achieved either at $l = \left\lfloor{\frac{(n-1)q}{n}}\right\rfloor$ or $l = \left\lceil{\frac{(n-1)q}{n}}\right\rceil = \left\lfloor{\frac{(n-1)q}{n}}\right\rfloor + 1$.
    We denote the values of $f$ in these two points with $F = f\left(\left\lfloor{\frac{(n-1)q}{n}}\right\rfloor\right)$ and $C = f\left(\left\lceil{\frac{(n-1)q}{n}}\right\rceil\right) = f\left(\left\lfloor{\frac{(n-1)q}{n}}\right\rfloor + 1\right)$.
    Using the binomial theorem we obtain
    \begin{align*}
        F - C = \sum_{i=0}^{n-1}\binom{n-1}{i}\left\lfloor{\frac{(n-1)q}{n}}\right\rfloor^i
        - \left(q- \left\lfloor{\frac{(n-1)q}{n}}\right\rfloor\right)\sum_{i=0}^{n-2}\binom{n-1}{i}\left\lfloor{\frac{(n-1)q}{n}}\right\rfloor^i.
    \end{align*}
    Now let us introduce $m$, such that $(n-1)q \equiv m \pmod{n}$, to rewrite the difference
    \begin{align*}
        F - C =& \left(\frac{(n-1)q - m}{n}\right)^{n-1}
        + \frac{n - q- m}{n} \sum_{i=0}^{n-2}\binom{n-1}{i} \left(\frac{(n-1)q - m}{n}\right)^i \\
        =& (n-k-1)\left(\frac{(n-1)q - m}{n}\right)^{n-2}
        +\frac{n - q -m}{n} \sum_{i=0}^{n-3} \binom{n-1}{i} \left(\frac{(n-1)q - m}{n}\right)^i.
    \end{align*}

    First let us explain that the coefficient $\frac{n - q- m}{n}$ is non-positive.
    Suppose for contradiction that $q < n - m$.
    Then $\frac{(n-1)q}{n} < n - m + 1 + \frac{m}{n}$.
    Since $(n-1)q \equiv m \pmod{n}$, it follows $\frac{(n-1)q}{n} \leq n - m + 1 + \frac{m}{n} - n = \frac{m}{n} - m - 1 < 0$
    but $\frac{(n-1)q}{n} < 0$ cannot hold for $n > 1$ and $q > 1$.
    We also see that $\frac{n - q- m}{n} = 0$ if and only if $\left\lceil{\frac{(n-1)q}{n}}\right\rceil = q$.
    In this case, $F - C > 0$ and therefore $\max_{i \in \{1,\dots, q-1\}} f(i) = F$.
    We immediately see that for $m = n-1$, $F < C$ since $q > 1$ always holds.
    To get a lower bound on $F - C$ when $m \leq \frac{n}{2}$ we apply lemma~\ref{lemma:nhalves} for $j=3, \dots, n$ consecutively.
    We obtain $F - C \geq (n - m -1) \prod_{i=2}^{n-1}\frac{n-i}{i} > 0$.
\end{proof}

Now we can compute the sizes of maximum four-letter non-overlapping codes.
The proof of Theorem~\ref{theorem:fourletter} reveals that they can indeed be generated using Construction~\ref{construction:blackburn}.

\begin{theorem}\label{theorem:fourletter}
    \[S(q,4) = \left[{\frac{3q}{4}}\right]^3 \left(q - \left[{\frac{3q}{4}}\right]\right),\]
    where $\left[x\right]$  is the rounding of $x$ to the closest natural number
    that rounds half down.
\end{theorem}
\begin{proof}
    Let $C$ be a maximum code that satisfies property~\eqref{eq:property}.
    Without loss of generality we can assume $\lvert R_1\rvert  \geq \lvert L_1\rvert $ (otherwise we can switch L's and R's and obtain a code of the same size),
    meaning that $R_3 = \emptyset $ and $\lvert L_3\rvert  = \lvert L_1\rvert  \lvert R_2\rvert  + \lvert L_2\rvert  \lvert R_1\rvert $.
    We also know that $\lvert R_2\rvert  = \lvert L_1\rvert  \lvert R_1\rvert  - \lvert L_2\rvert $,
    so the size of $C$ equals
    \begin{align*}
        \lvert C\rvert  =& \lvert L_1\rvert  \lvert R_3\rvert  + \lvert L_2\rvert  \lvert R_2\rvert  + \lvert L_3\rvert  \lvert R_1\rvert  \\
        =& \lvert L_2\rvert  \lvert R_2\rvert   + \lvert L_1\rvert  \lvert R_2\rvert  \lvert R_1\rvert  + \lvert L_2\rvert  \lvert R_1\rvert ^2  \\
        =& - \lvert L_2\rvert ^ 2 + \lvert L_2\rvert \lvert  R_1 \rvert ^2 + \lvert L_1\rvert  ^2 \lvert R_1\rvert ^2.
    \end{align*}

    Using the derivative test, we know that for fixed sets $L_1$ and $R_1$ with $\lvert R_1\rvert  \leq 2\lvert L_1\rvert $ the above function achieves its maximum
    when $\lvert L_2\rvert  = \frac{\lvert R_1\rvert  ^2}{2}$.
    A feasible partition $(L_1, R_1)$ with $\lvert R_1\rvert \leq 2\lvert L_1\rvert$ exists for every $q \geq 2$ (take $\lvert R_1\rvert = \lceil \frac{q}{2} \rceil$), so $\lvert C \rvert  \leq \frac{\lvert R_1\rvert ^4}{4} + \lvert R_1\rvert ^2 \lvert L_1\rvert ^2 \leq 8 \lvert L_1\rvert ^4 \leq \frac{8q^4}{3^4}$.
    If $\lvert R_1\rvert  > 2 \lvert L_1\rvert $, the maximum code size is achieved when $\lvert L_2\rvert  = \lvert L_1\rvert  \lvert R_1\rvert $.
    Then $C$ is constructed using Construction~\ref{construction:blackburn} with $k = 3$.
    A feasible partition $(L_1, R_1)$ with $\lvert R_1\rvert > 2\lvert L_1\rvert$ exists for every $q > 3$ (take $\lvert L_1 \rvert = 1$). 

    If $q = 2$ then all feasible partitions $(L_1,R_1)$ satisfy $\lvert L_1 \rvert = \lvert R_1 \rvert = 1$. Therefore $\lvert L_2 \rvert \in \{0,1\}$ and $\lvert L_2 \rvert^2 = \lvert L_2\rvert$, so
    \begin{align*}
        S(2,4) &= - \lvert L_2 \rvert^2 + \lvert L_2\rvert + 1 
        = 1 \\
        &= 1^3 \left(2-1\right) \\
        &= \left[\frac{3\cdot 2}{4}\right]^3\left(2 - \left[\frac{3\cdot 2}{4}\right]\right).
    \end{align*}
    If $q = 3$ then for all feasible partitions $(L_1,R_1)$ such that $\lvert R_1 \rvert \geq \lvert L_1 \rvert$, $\lvert L_1 \rvert = 1$ and $\lvert R_1 \rvert = 2$ holds. Therefore $\lvert R_1 \rvert = 2 \lvert L_1 \rvert$.    
    A feasible partition $L_2 = \left(L_1R_1\right)$ with $\lvert L_2 \rvert = \frac{\lvert R_1\rvert^2}{2}$ exists, so
    \begin{align*}
        S(3,4) &= 8 \lvert L_1 \rvert = 8 \\
        &= 2^3 \left(3 - 2 \right) \\
        &= \left[\frac{3\cdot 3}{4}\right]^3\left(3 - \left[\frac{3\cdot 3}{4}\right]\right).
    \end{align*}

    We are left to prove that $\frac{8q^4}{3^4} \leq \max_{i \in \{1, \dots, q-1\}} (q - i) i^3$ for $q \geq 4$.
    We use Proposition~\ref{proposition:blackburn_optimalsize} to divide the problem into four subproblems.\\
        (i) If $4 \mid q$, we already know that Blackburn's construction yields the optimum.
        Still we can check that $\frac{8q^4}{3^4} < \frac{27q^4}{4^4}$ or equivalently $2^{11} = 2048 < 2187 = 3^7$. \\
        (ii) If $ 3q \equiv 1 \pmod{4}$, we have to check that $\frac{\left(q + 1\right)\left(3q-1\right)^3}{4^4} > \frac{8q^4}{3^4}$
        or equivalently $\frac{139}{2048}q^4 + \frac{81}{2048}\left(-18q^2 + 8q - 1\right) > 0$.
        The two real roots of this polynomial are $q \approx - 3.4482$ and $q = 3$.
        A strong inequality therefore holds for $q \geq 4$. \\
        (iii) If $ 3q \equiv 2 \pmod{4}$, we have to check that $\frac{\left(q + 2\right)\left(3q-2\right)^3}{4^4} > \frac{8q^4}{3^4}$
        or equivalently $\frac{139}{2048}q^4 + \frac{81}{2048}\left(-18q^2 + 32q - 4\right) > 0$.
        The two real roots of this polynomial are $q \approx -3.9235$ and $q \approx 0.13528$, so the inequality holds.\\
        (iv) If $ 3q \equiv 3 \pmod{4}$,  $\frac{\left(q - 1\right)\left(3q+1\right)^3}{4^4} > \frac{8q^4}{3^4}$
        or equivalently $\frac{139}{2048}q^4 - \frac{81}{2048}\left(18q^2 + 8q + 1\right) > 0$.
        The two real roots of this polynomial are $q = -3$ and $q \approx 3.4482$, so the inequality holds for $q \geq 4$.
\end{proof}

\begin{theorem}\label{theorem:fourletter_n}
    \[N(q,4) = \begin{cases} 2 \binom{q}{\left[{\frac{3q}{4}}\right]} & q \geq 3 \\ 6 & q = 2. \end{cases}\]
\end{theorem}
\begin{proof}
    The proof of the Theorem~\ref{theorem:fourletter} showed that for $q \geq 3$ all maximum non-overlapping codes with $\lvert R_1 \rvert \geq \lvert L_1 \rvert$ are given by partitions $\left(L_1, R_1 \right)$ with $\lvert L_1 \rvert = \left[\frac{3q}{4}\right]$.
    There are $\binom{q}{\left[{\frac{3q}{4}}\right]}$ such partitions.
    Since $\lvert R_1 \rvert > \lvert L_1 \rvert$, all maximum non-overlapping codes satisfy Property~\ref{eq:property}.
    For each partition $(L_1,R_1)$ partitions $(L_2,R_2), (L_3,R_3)$ are uniquely determined.
    From Proposition~\ref{proposition:symmetry} we know that there are also $\binom{q}{\left[{\frac{3q}{4}}\right]}$ maximum non-overlapping codes with $\lvert R_1 \rvert < \lvert L_1 \rvert$.

    If $q = 2$ then every four-letter non-overlapping code is maximum as we know from Theorem~\ref{theorem:fourletter} that $S(2,4) = 1$. Hence $\lvert L_1 \rvert = \lvert R_1 \rvert = 1$ and $\lvert L_i \rvert + \lvert R_i \rvert = 1$ for $i \in \{2,3\}$. Although there are 8 choices for the selection of partitions, we obtain only 6 distinct codes $\{0111\}$, $\{0001\}$, $\{0011\}$, $\{1000\}$, $\{1110\}$, and $\{1100\}$. Each of the codes $\{0011\}$ and $\{1100\}$ is given by two collections of partitions.
\end{proof}

Unfortunately Construction~\ref{construction:blackburn} fails to generate a maximum five-letter non-overlapping code, as we will see later.

\section{Comparison to other constructions}\label{results}
At the moment we are not able to provide any simple formula to compute the size of a maximum
non-overlapping code with a larger codeword length.
However, we solved SQN for $5 \leq n \leq 30$ when $q = 2$, and $5 \leq n \leq 16$ when $3 \leq q \leq 6$ using a search algorithm that incorporates the observations from previous sections (see Appendix~\ref{appendix:algorithm}).
For each pair of parameter values $q$ and $n$, a single optimal solution of SQN was found up to the transformations of Propositions~\ref{proposition:symmetry}, \ref{proposition:evenq} and Corollary \ref{cor:1}.
We computed the full set of optimal solutions of SQN by applying those propositions, and $N(q,n)$ from Proposition~\ref{proposition:nqn}.
The results are presented in Tables~\ref{tab:max_sqn} and \ref{tab:max_binary_sqn}.

We checked that for $q=2, n \leq 16$ the computed values of $S(q,n)$ agree with the values obtained by Chee et al~\cite{Chee:2013}, and for $n > 16$ the determined values are not smaller than the largest code obtained by Constructions ~\ref{construction:levenshtein} and \ref{construction:bilotta}. For $16 < n \leq 30$ the gap between $S(q,n)$ and the size of the largest code given by these two constructions is given in Table~\ref{tab:max_binary_sqn}. Note that for this set of parameters Construction~\ref{construction:levenshtein} always yields larger codes than Construction~\ref{construction:bilotta}~\cite{Chee:2013}.

\begin{table}[h!tb]
    \begin{center}
	\begin{threeparttable}
        \begin{minipage}{\textwidth}
            \caption{$S(q,n)$ and $N(q,n)$ for $2 \leq q \leq 6$ and $3 \leq n \leq 16$.}
            \label{tab:max_sqn}
\begin{tabular}{ >{\raggedleft}p{1em}  >{\raggedleft}p{2em} >{\raggedleft}p{4em} >{\raggedleft}p{4em}>{\raggedleft} p{4em}>{\raggedleft} p{7em}>{\raggedleft\arraybackslash} p{7.8em}}
    \hline
    & & \multicolumn{5}{c}{$q$} \\\cline{3-7}
    $n$  &  &  2 &  3 & 4 & 5 & 6\\\hline
    3 & $S(q,n)$ & 1 & 4 & 9 & 18 & 32 \\
      & $N(q,n)$ &  4 & 6 &8 & 20 & 30\\\hline
    4 & $S(q,n)$ & 1 & 8 & 27 & 64 & 128 \\
      & $N(q,n)$ &  6 & 6 & 8 & 10 & 30\\\hline
    5 & $S(q,n)$ & 2 & 17 & 81 & 256 & 625 \\
    & $N(q,n)$ &  8 & 12& 8  & 10 & 12 \\\hline
    6 & $S(q,n)$ & 3 & 41 & 251 & 1024 & 3125 \\
    & $N(q,n)$ &  16 & 12& 24 & 10 & 12 \\\hline
    7 & $S(q,n)$ & 5 & 99 & 829 & 4181 & 15625\\
    & $N(q,n)$ &  48 & 12& 24 & 40 & 12 \\\hline
    8 & $S(q,n)$& 8 & 247 & 2753 & 17711 & 79244\\
    & $N(q,n)$ &  288 & 36 & 24 & 40 & 60\\\hline
    9 & $S(q,n)$& 14 & 656 & 9805 & 76816 & 411481\\
    & $N(q,n)$ & 1132 & 6& 24  & 60 & 60\\\hline
    10 & $S(q,n)$& 24 & 1792 & 34921 & 341792 & 2188243 \\
     & $N(q,n)$ & $2^{15}$  & 6& 24  & 60 & 120\\\hline
    11 & $S(q,n)$ & 44 & 4896 & 124373 & 1520800 & 11755857\\
     & $N(q,n)$ & $2^{26}$ & 6& 24  & 60 & 120\\\hline
    12 & $S(q,n)$ & 81 & 13376 & 446496 & 6817031 & 63281718 \\
     & $N(q,n)$ & $2^{46}$ & 6& 120  & 360 & 2040 \\\hline
    13 & $S(q,n)$ & 149 & 36544 & 1619604 & 31438129 & 350255809 \\
     & $N(q,n)$ & $2^{83}$ & 6& 120 & 40 & 120 \\\hline
    14 & $S(q,n)$ & 274 & 99840 & 5941181 & 146053729 & 19404900978 \\
     & $N(q,n)$ & $2^{151}$ & 6 & 240 & 40 & 120 \\\hline
    15 & $S(q,n)$ & 504 & 274384 & 21917583 & 678529303 & * $\geq 10755272317$ \\
     & $N(q,n)$ & $2^{276}$ & 24 & 240 & 40 & * \\\hline
    16 & $S(q,n)$& 927 & 759847 & 82990089 & * $\geq 3152278399$ & * $\geq 59811113295$ \\
     & $N(q,n)$ & $2^{506}$ & 24 & 8 & * & * \\
    \hline
    \end{tabular}
    \begin{tablenotes}\footnotesize
            \item[*]{The (exact) value was not determined due to the long expected runtime of the algorithm.}
    \end{tablenotes}
        \end{minipage}
        \end{threeparttable}
        \end{center}
\end{table}

\begin{table}[htb]
    \begin{center}
    	\begin{threeparttable}
        \begin{minipage}{\textwidth}
            \caption{$S(2,n)$ and $N(2,n)$ for $17 \leq n \leq 30$. The last column equals the difference between $S(2,n)$ and the largest non-overlapping code given by Construction~\ref{construction:levenshtein}.}
            \label{tab:max_binary_sqn}        
\begin{tabular}{ >{\raggedleft}p{2em}  >{\raggedleft}p{.25\textwidth} >{\raggedleft}p{.25\textwidth} > {\raggedleft\arraybackslash} p{.25\textwidth}}
    \hline
    $n$  & $S(2,n)$ & $N(2,n)$ & gap\\\hline
    17 & 1705 & $2^{930}$ & 0\\
    18 & 3160 & $2^{1677}$ & 24 \\
    19 & 5969 & $2^{3163}$ & 201\\
    20 & 11272 & $2^{5974}$ & 601\\
    21 & 21287 & $2^{11277}$ & 718\\
    22 & 40202 & $2^{21292}$ & 554\\
    23 & 76424 & $2^{39653}$ & 0\\
    24 & 147312 & $2^{76429}$ & 0\\
    25 & 283953 & $2^{147317}$ & 0\\
    26 & 547337 & $2^{283958}$  & 0\\
    27 & 1055026 & $2^{547342}$ & 0\\
    28 & 2033628 & $2^{1055031}$ & 0\\
    29 & 3919944 & $2^{2033633}$ & 0\\
    30 & *$\geq$ 7555935 & * & $\geq$ 0\\
    \hline
    \end{tabular}
    \begin{tablenotes}\footnotesize
            \item[*]{The (exact) value was not determined due to the long expected runtime of the algorithm.}
    \end{tablenotes}
        \end{minipage}
        \end{threeparttable}
        \end{center}
\end{table}

For $3 \leq q \leq 6$ and $5 \leq n \leq 16$, we compared our values with the optimal values of Construction~\ref{construction:levenshtein}, Wang and Wang's~\cite{Wang:2021} modification of Construction~\ref{construction:blackburn},
Wang and Wang's~\cite{Wang:2021} generalization of Construction~\ref{construction:bilotta}, and Construction~\ref{construction:barcucci} that have been included as a supplement to~\cite{Wang:2021}.
The gap between the optimal values of these constructions and the newly determined values of $S(q,n)$ are given in Table~\ref{tab:comparison}.
Our implementation of the algorithm always returned a value for $S(q,n)$ that is larger or equal to the size of the largest code given by any of those constructions.
Whenever we got a strictly better solution for some $q$ and $n$, we also checked that none of the maximum codes can be obtained using Construction~\ref{construction:blackburn}. This can be done by observing the set of optimal solutions of SQN and Proposition~\ref{proposition:blackburntowang}.

\begin{proposition}\label{proposition:blackburntowang}
    Let  $C = \bigcup_{i=1}^{n-1} \left(L_i R_{n-i} \right)$ be a maximum non-overlapping code given by Construction~\ref{construction:blackburn}
    for some value of $k$.
    If $R_{k+1} = \emptyset$, then $S = I^k = L_1^k$ and $C$ can be obtained by Wang and Wang's modification.
\end{proposition}
\begin{proof}
    Since $S \subseteq I^k$, $L_2=\cdots=L_k = \emptyset$ and $L_{k+1} =  \bigcup_{i=1}^{k} \left(L_i R_{k+1-i} \right) = \left(L_1 R_k \right)$.
    Similarly, $R_k = \left(L_1^{k-1} R_1 \right)$, so $L_{k+1} = \left(L_1^{k} R_1 \right)$ and $S = L_1^k$.
\end{proof}

\begin{sidewaystable}\footnotesize
    \begin{center}
        \begin{minipage}{\textwidth}
            \caption{Gap between $S(q,n)$ and the largest non-overlapping codes given by Construction~\ref{construction:levenshtein} (C1), Wang and Wang's generalization of Construction~\ref{construction:bilotta} (W2), Wang and Wang's modification of Construction~\ref{construction:blackburn} (W4), and Construction~\ref{construction:barcucci} (C5). For each parameter value $3 \leq q \leq 6$ and $5 \leq n \leq 16$ the first column states the gap between the code sizes, and the second column provides the list of constructions that yield the largest code.}
            \label{tab:comparison}
\begin{tabular}{ >{\raggedleft}p{1em} >{\raggedleft}p{6em}>{\raggedleft}p{6.5em}>{\raggedleft}p{.2em}>{\raggedleft}p{6em}>{\raggedleft}p{6em}>{\raggedleft}p{.2em}>{\raggedleft}p{6em}>{\raggedleft} p{6em}>{\raggedleft}p{.2em}>{\raggedleft} p{7em}>{\raggedleft\arraybackslash} p{6em}}
    \hline
    &  \multicolumn{11}{c}{$q$} \\\cline{2-12}
     & \multicolumn{2}{c}{3} & & \multicolumn{2}{c}{4} & & \multicolumn{2}{c}{5} & & \multicolumn{2}{c}{6}\\\cline{2-3}\cline{5-6}\cline{8-9}\cline{11-12}
    $n$ & gap & constructions && gap & constructions && gap & constructions && gap & constructions \\
    \hline
    5 & 1 & C1,W2,W4,C5 && 0 & C1,W4 && 0 & C1,W4 && 0 & C1,W4 \\
    6 & 5 & C5 && 8 &  C1,W4 && 0 & C1,W4 && 0 & C1,W4 \\
    7 & 11 & C1,W4 && 100 & C1,W4 && 85 & C1,W4 && 0 & C1,W4 \\
    8 & 7 & C1,W4 && 419 & C5 && 1327 &C1,W4 && 1119 & C1,W4 \\
    9 & 0 & C1,W4 && 1937 & C5 && 11280 & C1,W4 && 20856 & C1,W4 \\
    10 & 0 & C1,W4  && 6976  & C1,W4 && 62826 & W4 && 235118 & C1,W4 \\
    11 & 0 & C1,W4 && 18425 & C1,W4 && 300940 & W4 && 172849 & W4 \\
    12 & 0 & C1,W4 && 44817 & C1,W4 &&  1483667 & W4 && 8493622 & W4 \\
    13 & 0 & C1,W4 && 96723 & C1,W4 && 7922993 & C1,W4 && 50887361 & W4 \\
    14 & 0 & C1, W4 && 167501 & C1, W4 && 32512609 & C1, W4 && 305122418 & W4 \\
    15 & 1616 & C1, W4 && 27900 & C1, W4 && 130304279 & C1, W4 &&  $\geq$  1817210513 & W4 \\
    16 & 14631 & C1, W4 &&  0 & C1, W4 &&  $\geq$ 505213823 & C1, W4 & & $\geq$ 10972637519 & W4 \\
    \hline
    \end{tabular}
        \end{minipage}
        \end{center}
\end{sidewaystable}

In spite of the reduction of the optimization problem, finding all solutions to SQN still requires years of CPU time for some parameter values.
We therefore analysed the sets of optimal solutions to find some common properties.
The results lead to the following conjecture.

\begin{conjecture}\label{conjecture1}
    For every $q \geq 2$ and $n \geq 3$ there exists a maximum non-overlapping  code $C = \bigcup_{i < n} \left(L_i R_{n-i} \right) \in \mathcal{M}_{q,n}$ that satisfies $R_i = \emptyset$ for every $i > \frac{n}{q+1} + 1$.
\end{conjecture}

\begin{remark}
    If Conjecture~\ref{conjecture1} holds, then Blackburn's conjecture holds for $q \geq n$.
\end{remark}

We reduced the set of feasible solutions accordingly to find a large non-overlapping code for the parameter values, for which we could not solve the original optimization problem.
As shown in Table~\ref{tab:comparison} for all three parameter values a code was obtained that is much larger than the hitherto known ones. In the binary case, $n=30$, a code with exactly the same size as the one given by Construction~\ref{construction:levenshtein} was found.

\section*{Declarations}

\subsection*{Funding}

The research was partially supported by the scientific-research program P2-0359 and by the basic research project J1-50024, both financed by the Slovenian Research and Innovation Agency, and by the infrastructure program ELIXIR-SI RI-SI-2 financed by the
European Regional Development Fund and by the Ministry of Education, Science and Sport of Republic of Slovenia.

\subsection*{Competing interests}
The authors have no competing interests to declare that are relevant to the content of this article.

\subsection*{Availability of data and materials}
The set of optimal solutions of SQN that were determined in this study is available at: https://github.com/magdevska/nono-codes.

\subsection*{Code availability}
All code is available for download at: https://github.com/magdevska/nono-codes.

\appendix
\section{The algorithm for SQN}\label{appendix:algorithm}
Here we provide a pseudo-code of an algorithm that searches for a maximum non-overlapping code.
Procedure \textsf{compute\_size} is due to Theorem~\ref{prop:size}.
We obtain procedures \textsf{compute\_conditions} and \textsf{compute\_parameters} from Definition~\ref{definition:coefficients}. Theorem~\ref{theorem:maximum} raises procedure \textsf{determine\_upper\_half}.
Procedure \textsf{branch\_at\_level} combines Construction~\ref{construction:fimmel}, Propositions~\ref{proposition:symmetry} and \ref{proposition:evenq}, and Theorem~\ref{theorem:maximum}.
The solutions are compared with procedure \textsf{update\_maximum}. Procedure \textsf{solve\_sqn} initializes the optimal solution and starts the branching algorithm.

In order to obtain the results in a reasonable amount of time, our implementation uses multi-threading.
A depth of parallelization $d$ is given as parameter and should be selected depending on the available number of CPU cores and $q$.
We divide the branching stage into subproblems as follows.
If $i$ is smaller than $d$, then the branching procedure creates a new thread for each feasible value of $x_i$. After the child threads are joined, the parent thread selects the maximum values of its children $s^*$ and the corresponding solution set $C^*$.
If $i$ equals $d$, then the branching stage runs sequentially.

\begin{algorithm}
\caption{Non-overlapping code size}
    \begin{algorithmic}[1]
        \Require {$\{x_{1}, \dots, x_{i}\}$, $\{y_{1}, \dots, y_{i}\}$}
        \Procedure{compute\_size} {$i$}
        \State $s \leftarrow 0$
        \For{$ j \leftarrow 1\dots i- 1$}
            \State $s \leftarrow s + x_{j}y_{i-j}$
        \EndFor
        \EndProcedure
    \end{algorithmic}
\end{algorithm}

\begin{algorithm}
\caption{Computing the conditions $c_i$}
\begin{algorithmic}[1]
\Require {$\{x_{1}, \dots, x_{n-i}\}$, $\{y_{1}, \dots, y_{n-i}\}$, $\{\delta_{i+1}, \dots, \delta_{n-1}\}$, $\{p_1, \dots,p_{n-1-i}\}$, $i >  \frac{n}{2}$}
\Procedure {compute\_conditions} {$i$}
    \State $p_i \leftarrow $ compute\_parameters($i$)
    \State $c_i \leftarrow y_{n-i}  - x_{n-i}$
    \For {$j \leftarrow 1 \dots n-1-i$}
        \State $v \leftarrow 0$
        \For {$l \leftarrow 1 \dots j$}
            \State $v \leftarrow v + p_{lj}( y_l  - x_l )$
        \EndFor
        \State $c_i \leftarrow c_i + \left(\delta_{i+j} y_{n-i-j}  + (1-\delta_{i+j}) x_{n-i-j} \right) v$ 
    \EndFor
    \If {$c_i \geq 0$}
        \State $\delta_i \leftarrow 1$
    \Else
        \State $\delta_i \leftarrow 0$
    \EndIf
\EndProcedure
\Statex
\Procedure{compute\_parameters} {$i$}
    \State $p_{ii} \leftarrow 1$
    \For {$j \leftarrow i-1 \dots 1$}
        \State  $p_{ij} \leftarrow 0$
        \For {$k \leftarrow j\dots i-1$}
            \State  $p_{ij} \leftarrow p_{ij} + (\delta_{i+k} y_{i-k} + (1-\delta_{i+k}) x_{i-k})p_{ik}$
        \EndFor
    \EndFor
\EndProcedure
\end{algorithmic}
\end{algorithm}

\begin{algorithm}\label{algorithm:upperhalf}
    \caption{Determine $x_i$, $y_i$ for $i > \frac{n}{2}$}
    \begin{algorithmic}[1]
        \Require{$n$, $\{\delta_{\lfloor \frac{n}{2}\rfloor+1}, \dots, \delta_{n-1}\}$}
        \Procedure{determine\_upper\_half} {}
        \For{$i \leftarrow \lfloor \frac{n}{2}\rfloor+1 \dots n-1 $}
            \If {$\delta_i = 1$}
                \State $x_i \leftarrow $ compute\_size(i)
                \State $y_i \leftarrow 0$
            \Else
                \State $x_i \leftarrow 0 $
                \State $y_i \leftarrow $ compute\_size(i)
            \EndIf
        \EndFor
        \EndProcedure
    \end{algorithmic}
\end{algorithm}

\begin{algorithm}
\caption{Branch}
\begin{algorithmic}[1]
    \Require{$n \geq 3$, $q \geq 2$, $\{x_1,\dots,x_{i-1}\}$, $\{y_1,\dots,y_{i-1}\}$}
    \Procedure{Branch\_at\_level}{$i$}
        \If {$i = 1$}
            \For {$x_i \leftarrow 1\dots\lfloor\frac{n}{2}\rfloor$}
                \State $y_i \leftarrow q - x_i$
                \State compute\_conditions($n-1$)
                \State branch\_at\_level($i+1$)
            \EndFor
        \ElsIf {$i \leq \frac{n}{2}$}
            \State $s \leftarrow $ compute\_size($i$)
            \If {$x_1 = y_1, \dots, x_{i-1} = y_{i-1}$}
                \State max\_size $\leftarrow \lfloor\frac{s}{2}\rfloor$
            \Else
                \State max\_size $\leftarrow s$
            \EndIf
            \For{$x_i \leftarrow 0\dots$ max\_size}
                \State $y_i \leftarrow s - x_i$
                \State compute\_conditions($n-i$)
                \State branch\_at\_level($i+1$)
            \EndFor
        \Else
            \State determine\_upper\_half()
            \State update\_maximum()
        \EndIf
    \EndProcedure
\end{algorithmic}
\end{algorithm}

\begin{algorithm}
\caption{Update current maximum if necessary}
\begin{algorithmic}[1]
\Require{$S^*, C^*, n $}
\Procedure{update\_maximum}{$x,y$}
    \State $s \leftarrow $ compute\_size($n$)
    \If{$s = s^*$}
        \State $C^* \leftarrow C \cup (x,y)$
    \ElsIf{$s > s^*$}
        \State $C^* \leftarrow (x,y)$
        \State $s^* \leftarrow s$
    \EndIf
\EndProcedure
\end{algorithmic}
\end{algorithm}

\begin{algorithm}
\begin{algorithmic}[1]
\Procedure{solve\_sqn}{$q,n$}
    \State $s^* \leftarrow 0$
    \State branch\_at\_level(1)
\EndProcedure
\end{algorithmic}    
\end{algorithm}

\end{document}